\documentclass[11pt]{article}
\usepackage[margin=1in]{geometry}
\usepackage[T1]{fontenc}
\usepackage{lmodern}
\usepackage{microtype}
\usepackage{graphicx}
\usepackage[authoryear,round]{natbib}
\bibpunct{(}{)}{;}{a}{,}{,}
\usepackage{local_macros}
\usepackage{enumitem}
\usepackage{booktabs,tabularx}
\usepackage{caption}
\usepackage{placeins}
\usepackage{needspace}
\newcommand{\x}{{\mathbf{x}}}
\newcommand{\y}{{\mathbf{y}}}
\newcommand{\z}{{\mathbf{z}}}

\newcommand{\R}{{\mathbb{R}}}

\newcommand{\Exp}{{\ensuremath{\mathbb{E}}}}
\DeclareMathOperator*{\argmin}{arg\,min}
\DeclareMathOperator*{\argmax}{arg\,max}

\newcommand{\MAXKSAT}{{\text{MAX}-\textit{k}-\text{SAT}}}

\newcommand{\MAXSAT}{{\text{MAX}-\text{SAT}}}
\newcommand{\MAXKCSP}{{\text{MAX}-\textit{k}-\text{CSP}}}
\newcommand{\MAXSNP}{{\text{MAX}-\text{SNP}}}

\newcommand{\MAXCUT}{{\text{MAX}-\text{CUT}}}

\newcommand{\LEAP}{\textnormal{\textsc{Leap}}}
\newcommand{\OPT}{{\text{OPT}}}

\renewcommand{\P}{{\text{P}}}
\newcommand{\NP}{{\text{NP}}}
\newcommand{\APX}{{\text{APX}}}
\newcommand{\NPhard}{{\text{NP}-hard}}

\newcommand{\dgn}{\operatorname{dgn}}

\hypersetup{
  hypertexnames=false,
  hidelinks,
  pdftitle={Learning-Augmented Smooth Integer Programs with PAC-Learnable Oracles},
  pdfauthor={Hao-Yuan He, Jiang-Tian Xue, Ming Li and Zhi-Hua Zhou},
  pdfsubject={Preprint},
  pdfkeywords={smooth integer programs, learning-augmented algorithms, combinatorial optimization, approximation algorithms, PAC learnability}
}
\newenvironment{keywords}{%
  \par\smallskip\noindent\textbf{Keywords: }\ignorespaces
}{\par\medskip}

\title{Learning-Augmented Smooth Integer Programs with \\PAC-Learnable Oracles}
\author{%
  Hao-Yuan He \quad Jiang-Tian Xue \quad Ming Li \quad Zhi-Hua Zhou\\[0.5em]
  \small National Key Laboratory for Novel Software Technology\\
  \small School of Artificial Intelligence, Nanjing University\\
  \small \texttt{\{hehy, xuejt, lim, zhouzh\}@lamda.nju.edu.cn}
}
\date{}

\begin{document}
\maketitle
\begin{abstract}%
We study an \emph{end-to-end} learning-augmented framework for \emph{smooth integer programs}, a broad class of Boolean polynomial optimization problems including \MAXCUT{} and \MAXKSAT{}.
An oracle learned from sampled instances predicts a full Boolean assignment for each unseen instance.
For fixed-degree objectives, our polynomial-time pipeline, \LEAP{}, linearizes the objective around the predicted assignment, solves an error-budgeted linear program, and rounds the fractional solution.
We establish two guarantees.
First, for any prediction, the approximation loss of \LEAP{} degrades gracefully with prediction error: it is controlled by an \emph{error footprint} that quantifies how much mispredicted coordinates can perturb the objective's local coefficients.
Second, we bound the pseudo-dimension of the \LEAP{}-induced cost class, yielding a PAC-style guarantee for learning an oracle from a predictor class via empirical risk minimization on instances from an unknown distribution.
Together, these results give an end-to-end bound that separates the expected footprint loss of the best oracle in the class from the statistical error of learning.
In experiments on \MAXCUT{} and \MAXSAT{}, \LEAP{} substantially repairs imperfect predictions, empirical risk minimization selects effective oracle candidates, and the pipeline scales to application benchmarks with up to two million clauses.
\end{abstract}

\section{Introduction}\label{sec:intro}
Combinatorial optimization is a central task in computer science, operations research, and many scientific applications.
Canonical problems such as \MAXCUT{} and \MAXKSAT{} ask for the maximization of an objective over a discrete feasible set, and are typically \NPhard{}; many are also \APX-hard~\citep{Papadimitriou1991APX}, ruling out a polynomial-time approximation scheme unless \(\P=\NP\)~\citep{Arora1998Hardness}.
In many application settings, however, optimization instances are not encountered in isolation: similar instances are solved repeatedly, creating an opportunity to learn predictions from past data and use them to guide future optimization.
This is the premise of \emph{learning-augmented algorithms}, which use data-driven predictions while retaining provable performance guarantees~\citep{Lykouris2021LearningAugmented, mitzenmacher2022algorithms}.

In this work, we focus on \emph{smooth integer programs}: the maximization of an \(n\)-variate degree-\(d\) smooth polynomial \(p(\x)\) over Boolean assignments,
\begin{equation}\tag{d-IP}\label{eqn:d-IP-intro}
  \begin{array}{cl}
\displaystyle
    \max_{\x} & p(\x)\\
    \text{s.t.} & \x \in \{0,1\}^n.
  \end{array}
\end{equation}
Informally, smoothness bounds the coefficient of each degree-\(\ell\) monomial by \(O(n^{d-\ell})\).\footnote{The degree \(d\) is fixed; hidden constants may depend on fixed parameters, including density constants.}
This model was introduced by \citet{Arora1999PTAS} and is broad enough to capture \MAXSNP{} problems~\citep{Papadimitriou1991APX}, including canonical graph partitioning tasks such as \MAXCUT{} and Boolean constraint satisfaction problems such as \MAXKSAT{}.

\citet{Arora1999PTAS} developed PTASs for \emph{dense} smooth integer programs, where \(\OPT=\Omega(n^d)\), using exhaustive sampling, linear programming, and rounding.
A key insight in their analysis is that local coefficients estimated from a partial assignment can turn the nonlinear objective into a linear one.
Because an optimal assignment is unknown, their algorithm samples a small set of variables and enumerates all Boolean assignments on the sample, which takes \(n^{O(1/\eta^2)}\) time for a target precision \(\eta\).

More recently, \citet{bampis24PLA} incorporated partial predictions into this sampling-and-LP paradigm through a \emph{parsimonious} oracle.
Their oracle predicts labels only on a sampled set of variables, which replaces the enumeration step, and the guarantee degrades with the number of wrong sampled predictions.
Like that of \citet{Arora1999PTAS}, their guarantee applies to \emph{dense} objectives.
This result treats the oracle as an external source of advice, leaving open how the oracle itself is obtained.

Because both methods linearize the objective from a random sample of variables, their guarantees control the resulting approximation error through global density, and the bound of \citet{bampis24PLA} depends only on the number of mispredicted labels in the sample.
Such bounds treat all prediction errors alike, regardless of where they occur.
In combinatorial problems, however, prediction errors are not equally consequential.
In \MAXCUT{} on a star graph, for instance, flipping the center's label in an optimal assignment loses every cut edge, whereas flipping a leaf's label loses only one.
More generally, two predictions with the same number of incorrect variables can distort the linearized objective to very different degrees, depending on the structure of the instance.
Capturing this difference requires an analysis that tracks how individual prediction errors perturb the linearized objective, rather than merely counting how many errors are made.

To capture this structural sensitivity, we let the oracle predict a full Boolean assignment.
This oracle model aligns naturally with modern learning pipelines for graphs and formulas, which typically output a label for every variable~\citep{Selsam2019LBLMD, pan2025can}.
A full prediction determines the linearized objective at once, removing both the sampling step of the methods above and the enumeration step of \citet{Arora1999PTAS}.
Without sampling noise, the linearization error arises solely from mispredicted variables.
By decomposing the polynomial objective recursively, we can trace how these mispredictions perturb the local coefficients of the linear surrogate.

This leads to an \emph{error footprint} that measures their cumulative influence, yielding approximation guarantees that reflect where the errors occur and extend beyond the dense regime.
We further study how to learn such an oracle from sampled instances, establishing an end-to-end learning-augmented framework.
This end-to-end viewpoint leads to two natural questions that organize the paper:
\begin{enumerate}
  \item\label{q:approximation-from-prediction} \textit{Given a predicted assignment that may contain errors, can we turn it into an approximation algorithm whose guarantee degrades gracefully with the prediction error?}
  \item\label{q:learned-oracle-performance} \textit{Can one learn such a predictor from sampled instances, so that the resulting algorithm has a PAC-style performance guarantee?}
\end{enumerate}

\subsection{Framework and Main Results}
We answer these questions through a learning-augmented approximation framework, using \MAXKSAT{} as a concrete running example.
Consider a \MAXKSAT{} formula with clause set \(\mathcal C\) without duplicate clauses, where each clause contains exactly \(k\) literals on distinct variables.
For a clause \(C\), let \(C^+\) and \(C^-\) denote its positive and negative literal indices.
The number of satisfied clauses is the Boolean polynomial
\[
  p(\x)
  =
  \sum_{C\in\mathcal C}
  \left(
    1-\prod_{i\in C^-}x_i\prod_{i\in C^+}(1-x_i)
  \right).
\]
The product term is one exactly when clause \(C\) is falsified.
The objective is a degree-\(k\) smooth integer program: each degree-\(\ell\) coefficient is \(O(n^{k-\ell})\).

For \hyperref[q:approximation-from-prediction]{Q1}, we develop \LEAP{}, short for LEarning-augmented Approximation via linear Programming.
Given an instance \(\pi\), an oracle \(f\) predicts a full Boolean assignment \(f(\pi)=\widehat{\x}\in\{0,1\}^n\).
\LEAP{} uses this prediction as a reference point and proceeds in three steps.
\begin{enumerate}[label=(\arabic*), leftmargin=*]
  \item \textbf{Linearize.} Decompose the objective recursively into variables multiplied by lower-degree local coefficients, and evaluate these coefficients at \(\widehat{\x}\) to obtain a linear surrogate.
  For example, a clause \(x_1\vee x_2\) contributes \(x_1+x_2-x_1x_2=x_1(1-x_2)+x_2\) to \(p\).
  Evaluating the local coefficient \(1-x_2\) at \(\widehat{\x}\) gives the linear surrogate \(x_1(1-\widehat{x}_2)+x_2\).
  Mispredicting \(\widehat{x}_2\) changes this coefficient by one.
  \item \textbf{Relax and round.} For an error budget \(b\), optimize the surrogate by a linear program that allows each local coefficient to vary by as much as \(b\) mispredicted coordinates in \(\widehat{\x}\) could perturb it, and round the fractional solution to an integral assignment.
  Since \(p\) is multilinear, rounding the coordinates one at a time, each to the better of \(0\) and \(1\), does not decrease \(p\); we call this \emph{deterministic rounding}.
  \item \textbf{Select.} Since the prediction error is unknown, repeat step~(2) over error budgets and return the best rounded assignment.
\end{enumerate}

Given an oracle prediction \(\widehat{\x}\), let \(\z\) be the assignment returned by \LEAP{}.
Its loss relative to \(\OPT\) is controlled by the error footprint of the prediction.
\Cref{thm:error-footprint-guarantee} proves this for every multilinear objective; we state it here for \MAXKSAT{}.

\begin{theorem}[Informal version of \cref{thm:error-footprint-guarantee}]\label{thm:informal-error-footprint}
  Consider a \MAXKSAT{} instance with objective \(p\), and let \(\x^*\) be an optimal assignment.
  Let \(\varepsilon=\|\widehat{\x}-\x^*\|_1\) be the number of coordinates in which the oracle prediction \(\widehat{\x}\) differs from \(\x^*\).
  \LEAP{} returns an assignment \(\z\) satisfying
  \[
    p(\z)\ge \OPT-2\Phi_\varepsilon(\widehat{\x}).
  \]
  Here \(\Phi_\varepsilon(\widehat{\x})\) is the error footprint: it measures how much up to \(\varepsilon\) prediction errors around \(\widehat{\x}\) can change the local coefficients of the decomposed objective.
\end{theorem}

With a perfect prediction, \(\Phi_0(\widehat{\x})=0\), and \LEAP{} returns an optimal assignment.
LP duality can further sharpen this bound by an instance-dependent factor (\cref{sec:lp-refinement}).
We bound the error footprint in two ways.

A global bound, which uses only the smoothness scale of the objective, gives a \emph{near-dense} statement.
\LEAP{} returns an assignment satisfying
\[
  p(\z)\ge \OPT-O(n^{k-1}\varepsilon).
\]
If \(\OPT=\Omega(n^{k-1+\gamma})\) for a constant \(\gamma\in[0,1]\), the approximation ratio is \(1-O(\varepsilon/n^\gamma)\).

A local bound, which uses instance structure, gives a \emph{bounded-occurrence} statement by charging each prediction error to only \(O(\Delta)\) affected clauses.
If each literal occurs in at most \(\Delta\) clauses, then \LEAP{} returns an assignment satisfying
\[
  p(\z)\ge \OPT-O(\Delta\varepsilon).
\]
Since \(\OPT=\Omega(n)\) once variables that appear in no clause are removed, the approximation ratio is \(1-O(\Delta\varepsilon/n)\).
Crucially, because the relaxation sees prediction errors only through the local coefficients, an error that heavily degrades the objective need not perturb the coefficients if the decomposition order is chosen well.
Optimizing the variable order of the decomposition tightens this bound further (\cref{prop:structural-gain}) and can even eliminate the impact of high-degree errors entirely, as on star graphs (\cref{example:stars}).

Both statements extend beyond the dense regime \(\OPT=\Omega(n^k)\) studied by \citet{Arora1999PTAS} and \citet{bampis24PLA} to instances with \(\OPT=o(n^k)\).
For the near-dense regime, \citet{Fotakis2016CSP} give a sampling-based approximation scheme that runs in subexponential time; given a prediction, \LEAP{} runs in polynomial time, with a guarantee that depends on the prediction error.
The order of the local bound cannot be improved uniformly: unless \(\P=\NP\), no polynomial-time algorithm guarantees loss \(o(\Delta\varepsilon)\), even when \(\varepsilon\) is known (\cref{sec:local-order-optimality}).

We next turn from analyzing a given oracle prediction to selecting an oracle from sampled instances.
For \hyperref[q:learned-oracle-performance]{Q2}, we establish a PAC-style guarantee for oracle selection using the algorithm-selection framework of \citet{Gupta2017PACApproach}.
Instances are drawn from an unknown distribution, and each predictor \(f\) induces a \LEAP{} algorithm.
We choose a predictor by empirical risk minimization~(ERM) over the cost of the solutions returned by these induced algorithms, and bound the pseudo-dimension of the resulting cost class.
Let \(\mathcal{P}(f)\) denote the expected objective value of the solution that \LEAP{} returns with predictor \(f\), and let \(H\) upper-bound objective values.
For \MAXKSAT{}, \(\mathcal{P}(f)\) is the expected number of satisfied clauses, and \(H\) may be taken as an upper bound on the number of clauses.

The following end-to-end bound combines an ERM guarantee for the induced cost class with the footprint theorem.

\begin{theorem}[Informal version of \cref{thm:erm-learnability}]\label{thm:informal-learnability}
  Fix \(\alpha>0\) and \(\delta\in(0,1]\).
  Let \(\mathcal{F}\) be a predictor class, and let \(d_{\mathcal{F}}\) be the sum of the VC dimensions of the \(n\) classes formed by its coordinate functions.
  Suppose ERM selects \(\hat f_{\mathrm{ERM}}\in\mathcal{F}\) from \(m\) sampled instances, where
  \[
    m
    =
    \Omega\left(\left(\frac{H}{\alpha}\right)^2
      \left(d_{\mathcal{F}}\log(e d_{\mathcal{F}})+\log\frac{1}{\delta}\right)\right).
  \]
  Then, with probability at least \(1-\delta\),
  \[
    \mathcal{P}(\hat f_{\mathrm{ERM}})
    \ge
    \mathcal{P}^*
    -\Lambda_{\mathcal F}^{\Phi}
    -2\alpha.
  \]
  Here \(\mathcal{P}^*\) denotes the expected optimum, and \(\Lambda_{\mathcal F}^{\Phi}\) is the footprint loss of the best predictor in \(\mathcal F\), defined formally in \cref{sec:learnability}.
\end{theorem}

The bound separates this footprint loss from the statistical error \(2\alpha\) of ERM.
For \MAXKSAT{}, the two footprint bounds above give
\(\Lambda_{\mathcal F}^{\Phi}=O(n^{k-1}\varepsilon_{\mathcal F})\) under the global smoothness bound and
\(\Lambda_{\mathcal F}^{\Phi}=O(\Delta\varepsilon_{\mathcal F})\) under bounded occurrence, where \(\varepsilon_{\mathcal F}\) is the best expected prediction error achieved by a predictor in \(\mathcal F\).

\subsection{Our Contributions}

Our main contributions are as follows.
\begin{enumerate}
  \item
  \emph{A learning-augmented approximation pipeline for smooth integer programs.}
  We introduce \LEAP{}, which turns a predicted assignment into a linear surrogate of a smooth polynomial objective, optimizes it by linear programming, and rounds the result, using one oracle query and polynomial time (\cref{sec:approach}).

  \item\emph{Error-footprint guarantees and their optimal order.}
  We show that the loss of \LEAP{} is at most twice the error footprint of the prediction, and that LP duality can sharpen this bound.
  The footprint admits a global bound from smoothness and a local bound from instance structure, and both extend beyond the dense regime (\cref{sec:approach,sec:applications}).
  Accounting for the structure of the instance more finely, through the order of the decomposition, tightens the local bound further.
  Unless \(\P=\NP\), no polynomial-time algorithm improves the order of this local bound (\cref{sec:local-order-optimality}).

  \item\emph{PAC-style guarantees for oracle selection.}
  We bound the pseudo-dimension of the cost class induced by \LEAP{}, which gives a sample-complexity guarantee for ERM-based oracle selection and an end-to-end bound that separates the footprint loss of the best oracle from the statistical error (\cref{sec:learnability}).

  \item\emph{Empirical evaluation.}
  We evaluate \LEAP{} on \MAXCUT{} and \MAXSAT{} instances.
  The experiments show that \LEAP{} substantially reduces the objective gap of imperfect predictions and that selection by empirical induced cost identifies effective oracle candidates.
  We also compare \LEAP{} with state-of-the-art solvers on large \MAXSAT{} application benchmarks (\cref{sec:experiments}).
\end{enumerate}

\subsection{Roadmap}
We discuss related work in \cref{sec:related-works}.
\Cref{sec:approach} presents \LEAP{} and its error-footprint analysis.
\Cref{sec:applications} applies the framework to \MAXCUT{} and \MAXKSAT{} and shows that the resulting local bounds have optimal order.
\Cref{sec:learnability} establishes the PAC-style guarantee for oracle selection.
\Cref{sec:experiments} evaluates \LEAP{} on \MAXCUT{} and \MAXSAT{} instances.
\Cref{sec:conclusion} concludes the paper.
The appendix contains supporting material, including extended related work, additional experimental details, omitted proofs, and the smoothness verification for \MAXKCSP{}, to which \MAXSNP{} problems reduce.
\section{Related Work}\label{sec:related-works}
We position the paper against three lines of work: smooth integer programming and dense CSP approximation, learning-augmented approximation algorithms, and PAC-style learning of algorithmic decisions.

\paragraph{Smooth Integer Programs and Prediction-Guided Relaxations.}
Smooth integer programs were introduced by \citet{Arora1999PTAS} as a polynomial framework for dense constraint satisfaction and related \MAXSNP{} problems.
\citet{Fotakis2016CSP} extended this sampling approach to the near-dense scale \(\Omega(n^{k-1+\gamma})\), called almost-sparse in their work.
More recently, \citet{bampis24PLA} incorporated predictions into the sampling-and-LP framework through a parsimonious oracle.
Following the same intuition, our framework uses the oracle prediction as a reference point for constructing the local quantities in the relaxation.
The main technical question is how prediction errors affect the relaxation and the rounded solution.
Our analysis answers this question by charging prediction errors to the components of the decomposed objective that they can influence, yielding both global smoothness consequences and local consequences when the instance structure permits them; the local consequences have optimal order among polynomial-time algorithms.

\paragraph{Learning-Augmented Approximation Algorithms.}
For approximation problems, prior work on learning-augmented algorithms has studied several advice models.
In label-advice models, the algorithm receives labels that are correct with probability slightly above \(1/2\), often under independence or pairwise-independence assumptions~\citep{Cohen2024MaxCut, ghoshal2025constraint}.
In subset-advice models, the optimal labels of a random subset of variables are revealed~\citep{Cohen2024MaxCut, Attias2025LAA-SAT}.
In edge-advice models for graph problems, each edge \((i,j)\) provides local predictions about \(x_i^*\) and \(x_j^*\)~\citep{Aamand2025Improved}.
These models yield prediction-sensitive guarantees for \MAXCUT{}, \MAXSAT{}, and related CSPs.
Our framework studies smooth integer programs with a full oracle prediction \(\widehat{\x}\).
The guarantee does not impose independence assumptions on oracle errors; its loss is measured through \(\|\widehat{\x}-\x^*\|_1\) and the error footprint.
A more detailed comparison is given in \cref{sec:detailed-approx-comparison}.

\paragraph{Learning Algorithmic Decisions from Data.}
A complementary line of work studies how to learn algorithmic decisions from instance distributions.
\citet{Gupta2017PACApproach} introduced a PAC framework for algorithm selection, where empirical performance over sampled instances is used to choose from a family of algorithms.
The central statistical requirement is a bounded pseudo-dimension for the induced real-valued cost class, which gives uniform convergence from empirical cost to expected cost.
This perspective has been developed further for branch-and-bound, algorithm configuration, and related families of algorithmic decisions~\citep{Balcan2018LearningToBranch, Balcan2017Learning, Balcan2021SampleComplexity, Khodak2022LAA-Learning, bartlett2022generalization, Balcan2024LAA-Learning}.
We use this perspective to study oracle selection.
Each predictor \(f\) induces a \LEAP{} algorithm, and ERM selects a predictor by the cost of the induced algorithm.
Combining the pseudo-dimension bound with the error-footprint guarantee gives an end-to-end learnability guarantee for learning-augmented smooth integer programs.
\section{Approach}\label{sec:approach}

This section develops the formal construction of \LEAP{} and proves its approximation guarantee, addressing \hyperref[q:approximation-from-prediction]{Q1}.
We first recall smooth Boolean polynomial objectives and their recursive decompositions.
We then study quadratic objectives as a motivating special case, where the surrogate objective and the error footprint have their simplest form.
We use this case to formulate the general footprint relaxation for fixed-degree smooth objectives, prove its approximation guarantee, and sharpen the guarantee through LP duality.
We close by deriving global and local bounds on the error footprint, and show how a suitable decomposition order improves the guarantee.
The main notation is summarized in \cref{tab:notation}.

\begin{table}[!t]
\centering
\small
\begin{tabular}{p{0.18\linewidth}p{0.72\linewidth}}
\hline
Symbol & Meaning \\
\hline
\(n,d\) & Number of Boolean variables and polynomial degree. \\
\(p,\beta\) & Smooth polynomial objective and its smoothness parameter. \\
\(\x^*,\OPT\) & An optimal assignment and the optimal value \(p(\x^*)\). \\
\(\widehat{\x}\) & Oracle prediction. \\
\(\varepsilon\) & Prediction error \(\|\widehat{\x}-\x^*\|_1\). \\
\(b\) & Error budget of the footprint relaxation. \\
\(\y_b,\z_b\) & LP solution and its rounding at budget \(b\); \(\z\) is the final output. \\
\(\Phi_b(\widehat{\x})\) & Error footprint at prediction \(\widehat{\x}\) and budget \(b\). \\
\(\sigma,L_j^\sigma\) & Variable order of the decomposition and ordered load of coordinate \(j\). \\
\(\dgn\) & Degeneracy of a graph or formula. \\
\(\Pi,\pi\) & Instance space and an instance. \\
\(\mathcal F,f\) & Predictor class and a predictor \(f:\Pi\to\{0,1\}^n\). \\
\hline
\end{tabular}
\normalsize
\caption{Main notation.}
\label{tab:notation}
\end{table}

\subsection{Preliminaries}
We recall the smooth polynomial model and the recursive decomposition used to construct the relaxation.

\begin{definition}[$\beta$-smooth]\label{def:smooth-polynomial}
  An \(n\)-variate polynomial \(p(\x)\) of degree at most \(d\) is \emph{\(\beta\)-smooth} if there exists a constant \(\beta > 0\) such that for every \(0 \leq \ell \leq d\), the absolute value of the coefficient of any degree-\(\ell\) monomial in the expansion of \(p(\x)\) is bounded by \(\beta n^{d-\ell}\).
\end{definition}

Many fundamental optimization problems have smooth polynomial objective formulations. We recall two examples used throughout the paper.

\begin{example}\label{example: maxcut}
  Given an undirected graph \(G=(V,E)\), \MAXCUT{} partitions \(V\) to maximize the number of crossing edges.
  We define the objective over \(\x \in \{0,1\}^n\), with \(x_i=1\) denoting that vertex \(i\) is in the subset, as
  \[
    p(\x)=\sum_{\{i,j\}\in E}\bigl(x_i(1-x_j)+x_j(1-x_i)\bigr),
  \]
  which is \(2\)-smooth.
\end{example}

\begin{example}\label{example: maxksat}
  Given a formula in conjunctive normal form without duplicate clauses, where each clause has exactly \(k\) literals on distinct variables, \MAXKSAT{} maximizes the number of satisfied clauses. The objective function is
  \[
    p(\x) = \sum_{C \in \mathcal{C}}  \left(1 - \prod_{i \in C^{-}} x_i \prod_{i \in C^{+}} (1 - x_i)  \right),
  \]
  where \(C^+\) and \(C^-\) denote indices of positive and negative literals in clause \(C\). This degree-\(k\) polynomial is \(4^k\)-smooth.
\end{example}

Further smoothness verification for \MAXKCSP{} is given in \cref{sec:maxcsp-smoothness}.

\paragraph{Multilinear reduction.}
Since \(x_i^2=x_i\) on \(\{0,1\}^n\), every smooth objective can be reduced to multilinear form before the relaxation is built.
Replacing every power \(x_i^{e}\) with \(e\ge1\) by \(x_i\) and collecting like monomials turns a degree-\(d\), \(\beta\)-smooth polynomial \(p\) into a multilinear polynomial \(\tilde p\) of degree at most \(d\) that agrees with \(p\) on \(\{0,1\}^n\).
A degree-\(\ell\) monomial of \(\tilde p\), for \(\ell\ge1\), collects at most \(\binom{d}{\ell}\le2^d\) monomials of \(p\), each of degree at least \(\ell\) and hence with coefficient at most \(\beta n^{d-\ell}\), so \(\tilde p\) is \(2^d\beta\)-smooth with respect to the degree bound \(d\).
The reduction preserves all Boolean objective values and optimal assignments, so we work from now on with the reduced multilinear objective, still denoted by \(p\), and compute all footprint quantities from it.

The main algebraic tool is a recursive decomposition of \(p\).
The root is indexed by \(\emptyset\), with \(p_{\emptyset}=p\).
For a node \(I\), its children are indexed by \((I,j)\), \(j\in[n]\).

\begin{lemma}\label{lemma:decomposition-of-polynomials}
  Let \(p(\x)\) be an \(n\)-variate \(\beta\)-smooth polynomial of degree \(d\).
  For every node \(I\) whose component \(p_I\) has positive degree, there exist a constant \(c_I\) and child components \(p_{I,j}\) such that
  \[
    p_I(\x)=c_I+\sum_{j\in[n]}x_jp_{I,j}(\x).
  \]
  A component \(p_I\) at depth \(t=|I|\) has degree at most \(d-t\), and each of its degree-\(\ell\) coefficients has magnitude at most \(\beta n^{d-t-\ell}\), including when \(p_I\) is constant.
\end{lemma}

\begin{proof}
Assigning each nonconstant monomial of \(p_I\) to one of its variables, for instance the first in a fixed ordering, gives such a recursive decomposition (\cref{fig:decomposition-tree}): a degree-\(\ell\) monomial of \(p_I\) becomes a degree-\((\ell-1)\) monomial of a child at depth \(t+1\) with the same coefficient.
Distinct monomials assigned to the same variable leave distinct quotients, so coefficients do not accumulate, and induction on the depth gives the stated bounds.
\end{proof}

Fix an order \(\sigma\) of the variables.
We call the decomposition that assigns every nonconstant monomial to its \(\sigma\)-first variable at every level the \emph{\(\sigma\)-decomposition}.
It factors each monomial along its variables in the order \(\sigma\), and \cref{sec:footprint-budget-analysis} shows how the choice of \(\sigma\) controls the local footprint bound.

\begin{figure}[t]
\centering
\begin{tikzpicture}[
  x=1cm,
  y=1cm,
  every node/.style={font=\footnotesize},
  comp/.style={draw=black!65, rounded corners=1pt, fill=black!3, inner sep=3pt, minimum width=1.35cm},
  child/.style={draw=black!60, rounded corners=1pt, fill=black!2, inner sep=3pt, minimum width=1.35cm},
  grand/.style={draw=black!55, rounded corners=1pt, dashed, fill=black!1, inner sep=3pt, minimum width=1.35cm},
  edge/.style={-Latex, thin, draw=black!60},
  ellipsis/.style={font=\footnotesize, text=black!65}
]
  \node[comp] (root) at (0,0) {\(p_I\)};

  \node[child] (c1) at (-2.15,-1.15) {\(p_{I,j}\)};
  \node[child] (c2) at (2.15,-1.15) {\(p_{I,k}\)};
  \node[ellipsis] at (0,-1.15) {\(\cdots\)};

  \node[grand] (g11) at (-3.55,-2.55) {\(p_{I,j,\ell}\)};
  \node[grand] (g12) at (-0.9,-2.55) {\(p_{I,j,m}\)};

  \node[grand] (g21) at (0.9,-2.55) {\(p_{I,k,\ell}\)};
  \node[grand] (g22) at (3.55,-2.55) {\(p_{I,k,m}\)};
  \node[ellipsis] at (-2.15,-2.55) {\(\cdots\)};
  \node[ellipsis] at (2.15,-2.55) {\(\cdots\)};

  \draw[edge] (root) -- (c1);
  \draw[edge] (root) -- (c2);
  \draw[edge] (c1) -- (g11);
  \draw[edge] (c1) -- (g12);
  \draw[edge] (c2) -- (g21);
  \draw[edge] (c2) -- (g22);
\end{tikzpicture}
\caption{
A schematic view of the recursive decomposition.
Each node represents a component \(p_I\).
}
\label{fig:decomposition-tree}
\end{figure}

\subsection{A Quadratic Warm-Up}\label{sec:quadratic-warm-up}
The error-footprint construction takes its simplest form for quadratic objectives.
Let \(p\) be an \(n\)-variate degree-\(2\), \(\beta\)-smooth polynomial, decomposed by \cref{lemma:decomposition-of-polynomials} as
\[
  p(\x)=c+\sum_{i\in[n]}x_ip_i(\x),
  \qquad
  p_i(\x)=c_i+\sum_{j\in[n]}x_jp_{i,j},
\]
where the second-level coefficients \(p_{i,j}\) are constants.
Given an oracle prediction \(\widehat{\x}\), evaluating each local coefficient at \(\widehat{\x}\) gives the linear surrogate \(q(\x)=c+\sum_{i\in[n]}x_ip_i(\widehat{\x})\).

How much can prediction errors distort these local coefficients?
A misprediction at coordinate \(j\) changes \(p_i\) by at most \(|p_{i,j}|\).
If \(\mathbf u\in[0,1]^n\) is a fractional certificate for the mispredicted coordinates, then \(p_i\) can vary within the tolerance \(\tau_i(\mathbf u)=\sum_{j}|p_{i,j}|u_j\), and coordinate \(j\) carries the footprint weight \(F_j=\sum_{i}|p_{i,j}|\) across all local coefficients.
Because the coefficients \(p_{i,j}\) are constants, these weights do not depend on \(\widehat{\x}\) in the quadratic case.
For a budget \(b\), the quadratic error footprint
\[
  \Phi^{(2)}_{b}
  =
  \max_{\substack{\mathbf u\in[0,1]^n\\ \sum_j u_j\le b}}\qquad
  \sum_{j\in[n]}F_ju_j
\]
bounds the total change that up to \(b\) prediction errors can cause in the local coefficients.
The footprint relaxation maximizes the surrogate while keeping every local coefficient within its tolerance:
\begin{equation}\tag{2-LP}\label{optimization:quadratic-footprint-lp}
  \begin{aligned}
    \max_{\y,\mathbf u} &\quad c+\sum_{i\in[n]}y_ip_i(\widehat{\x})\\
    \text{s.t.} &\quad |p_i(\y)-p_i(\widehat{\x})|\le\tau_i(\mathbf u) \quad \forall i\in[n],\\
    &\quad \sum_{j\in[n]}u_j\le b,\qquad \y,\mathbf u\in[0,1]^n .
  \end{aligned}
\end{equation}

The relaxation links the surrogate to the true objective in two steps.
First, every feasible \((\y,\mathbf u)\) has linearization error at most the footprint:
\[
  |p(\y)-q(\y)|
  \le\sum_{i\in[n]}y_i|p_i(\y)-p_i(\widehat{\x})|
  \le\sum_{i\in[n]}\tau_i(\mathbf u)
  =\sum_{j\in[n]}F_ju_j
  \le\Phi^{(2)}_{b}.
\]
Second, if \(\widehat{\x}\) differs from an optimal assignment \(\x^*\) on a set \(S\) of \(\varepsilon\) coordinates, then \((\x^*,\mathbf 1_S)\) is feasible at budget \(\varepsilon\), so an optimal LP solution \(\y\) at this budget satisfies
\[
  p(\y)
  \ge q(\y)-\Phi^{(2)}_\varepsilon
  \ge q(\x^*)-\Phi^{(2)}_\varepsilon
  \ge p(\x^*)-2\Phi^{(2)}_\varepsilon .
\]

\subsection{Footprint Relaxation and \LEAP{}}\label{sec:general-footprint-relaxation}
For general degree \(d\), we apply the same bookkeeping to a fixed recursive decomposition tree of \(p\) as in \cref{lemma:decomposition-of-polynomials}.
We write \(\mathcal T\) for the set of generated nodes, use \(\emptyset\) for the root, and constrain the non-root nodes whose components have positive degree,
\[
  \mathcal I=\{I\in\mathcal T: I\ne\emptyset,\ \deg(p_I)>0\}.
\]
Unlike the quadratic case, a component \(p_I\) may itself be nonlinear, so the relaxation works with its one-step linear surrogate at \(\widehat{\x}\),
\[
  q_I(\y)=c_I+\sum_{j\in[n]}y_jp_{I,j}(\widehat{\x}).
\]
For \(\mathbf u\in[0,1]^n\), the tolerance of a node, the footprint weight of a coordinate, and the error footprint at budget \(b\) are
\[
  \tau_I(\mathbf u)=\sum_{j\in[n]}|p_{I,j}(\widehat{\x})|u_j,
  \qquad
  F_j(\widehat{\x})=\sum_{I\in\mathcal I}|p_{I,j}(\widehat{\x})|,
  \qquad
  \Phi_{b}(\widehat{\x})
  =
  \max_{\substack{\mathbf u\in[0,1]^n\\ \sum_j u_j\le b}}\
  \sum_{j\in[n]}F_j(\widehat{\x})u_j .
\]
The guarantees below evaluate the footprint at the prediction error \(b=\varepsilon\).

The footprint relaxation is
\begin{equation}\tag{d-LP}\label{optimization:d-LP-oracle}
  \begin{aligned}
    \max_{\y,\mathbf u} \quad& c_{\emptyset}+\sum_{j\in[n]}y_jp_j(\widehat{\x})\\
    \text{s.t.} \quad& |q_I(\y)-p_I(\widehat{\x})|\le\tau_I(\mathbf u)\quad \forall I\in\mathcal I,\\
    & \sum_{j\in[n]}u_j\le b,\\
    & \y\in[0,1]^n,\quad \mathbf u\in[0,1]^n .
  \end{aligned}
\end{equation}
For \(d=2\), every \(q_I\) coincides with \(p_I\), and \cref{optimization:d-LP-oracle} reduces to \cref{optimization:quadratic-footprint-lp}.
Different valid decompositions yield different footprint weights, while the guarantee below applies to any fixed decomposition satisfying \cref{lemma:decomposition-of-polynomials}.

Because \(p\) is multilinear, a fractional solution of the footprint relaxation can be rounded deterministically without decreasing the objective value.

\begin{lemma}[Deterministic rounding for multilinear objectives]\label{lemma:deterministic-rounding-multilinear}
Let \(p(\x)\) be a multilinear polynomial. For any fractional solution \(\y \in [0,1]^n\), deterministic rounding yields an integral vector \(\z \in \{0,1\}^n\) such that
\[
  p(\z)\ge p(\y).
\]
\end{lemma}

\begin{proof}
Round the variables one at a time.
When all variables except \(x_i\) are fixed, multilinearity makes \(p\) affine in \(x_i\).
An affine function on \([0,1]\) is maximized at an endpoint, so choosing the better value in \(\{0,1\}\) does not decrease the objective.
Repeating this for all coordinates gives \(p(\z)\ge p(\y)\).
\end{proof}

A complete \LEAP{} protocol is given in \cref{algorithm: pipeline}.
Its rounding step is the deterministic rounding of \cref{lemma:deterministic-rounding-multilinear}.
Since the prediction error is unknown, the algorithm scans all budgets in \(\{0,\ldots,n\}\) and returns the best rounded candidate.

\begin{algorithm}[t]
  \caption{\LEAP{}: LEarning-augmented Approximation via linear Programming}
  \label{algorithm: pipeline}
  \DontPrintSemicolon
  \SetKwInOut{Input}{Input}
  \SetKwInOut{Output}{Output}

  \Input{Problem instance \(\pi \in \Pi\), oracle \(f\)}
  \Output{Integer solution \(\z\)}

  Obtain oracle prediction \(\widehat{\x} \leftarrow f(\pi)\)\;
  Let \(p\) be the objective polynomial of \(\pi\), replace powers \(x_i^e\) with \(x_i\) for \(e\ge1\), and collect like monomials\;
  Fix a recursive decomposition tree \(\mathcal T\) for the resulting multilinear polynomial \(p\), such as the \(\sigma\)-decomposition for a variable order \(\sigma\)\;

  \For{\(b \in \{0, 1, \dots, n\}\)}{
    Construct \cref{optimization:d-LP-oracle} at budget \(b\), and solve it to obtain a fractional solution \((\y_b,\mathbf{u}_b)\)\;
    Obtain \(\z_b\) using deterministic rounding\;
    Yield \(\z_b\) and \(p(\z_b)\)\;
  }
  \Return{the candidate \(\z_b\) with the largest evaluated objective value}
\end{algorithm}

\subsection{Error-Footprint Guarantee}\label{sec:error-footprint-guarantee}

We analyze the approximation quality of \LEAP{} by first evaluating a fixed error budget \(b\) and then accounting for budget selection.
To understand the loss of the LP solution, and hence of its rounded assignment, we decompose it into two terms: the distortion caused by linearizing the objective, which can enlarge the gap, and the advantage of the LP solution on the surrogate objective, which reduces it.

Formally, write \(q=q_{\emptyset}\) for the surrogate objective of \cref{optimization:d-LP-oracle} and \(\zeta=p-q\) for its linearization error, which can take either sign.
For a budget \(b\in[0,n]\), let \(P_b\subseteq[0,1]^n\) denote the projection of the feasible region of \cref{optimization:d-LP-oracle} at budget \(b\) onto the \(\y\)-coordinates.
Let \(\y_b\) be the optimal solution returned by the LP solver at budget \(b\), with surrogate value \(Q_b=q(\y_b)=\max_{\x\in P_b}q(\x)\), and let \(\z_b\) be its deterministic rounding.
For any feasible point \(\x\in P_b\), we define the \emph{surrogate gain} of \(\y_b\) over \(\x\) as \(g_b(\x)=Q_b-q(\x)\ge0\).
Whenever an optimal assignment \(\x^*\) lies in \(P_b\), the identity \(p=q+\zeta\) yields the error decomposition
\begin{equation}\label{eqn:loss-decomposition}
  \OPT-p(\y_b)
  =
  \underbrace{\bigl[\zeta(\x^*)-\zeta(\y_b)\bigr]}_{\text{linearization discrepancy}}
  -
  \underbrace{g_b(\x^*)}_{\text{surrogate gain}}.
\end{equation}
The first term is the net change in linearization error, and the second is the margin achieved by the LP on the surrogate.
We first bound the first term using only \(g_b(\x^*)\ge0\), and \cref{sec:lp-refinement} then uses LP duality to recover part of the second.
\begin{lemma}[Footprint control]\label{lemma:footprint-control}
For every budget \(b\) and every feasible \((\y,\mathbf u)\) of \cref{optimization:d-LP-oracle} at budget \(b\),
\[
  |\zeta(\y)|\le\sum_{j\in[n]}F_j(\widehat{\x})u_j\le\Phi_b(\widehat{\x}).
\]
Moreover, if \(\x^*\in\{0,1\}^n\) and \(\varepsilon=\|\widehat{\x}-\x^*\|_1\), then \(\x^*\in P_\varepsilon\) and \(|\zeta(\x^*)|\le\Phi_\varepsilon(\widehat{\x})\); in particular, \(\x^*\in P_b\) for every \(b\ge\varepsilon\).
\end{lemma}

The proof, given in \cref{app:footprint-guarantee}, proceeds by induction over the decomposition tree: the linearization error of each component is at most the total tolerance of its constrained descendants, and these tolerances sum to \(\sum_jF_j(\widehat{\x})u_j\) at the root.
For the second claim, the indicator vector of the mispredicted coordinates certifies that \(\x^*\) is feasible at budget \(\varepsilon\).

The surrogate gain is nonnegative, so dropping it from \cref{eqn:loss-decomposition} already gives the main guarantee.

\begin{theorem}[Error-footprint guarantee]\label{thm:error-footprint-guarantee}
Let \(p\) be a multilinear objective, let \(\x^*\) be an optimal assignment, and let \(\widehat{\x}\) be the oracle prediction.
Set \(\varepsilon=\|\widehat{\x}-\x^*\|_1\).
For every budget \(b\) with \(\x^*\in P_b\),
\[
  \OPT-p(\z_b)
  \le
  \Phi_\varepsilon(\widehat{\x})+\Phi_b(\widehat{\x}).
\]
Consequently, \LEAP{} with deterministic rounding returns an integral assignment \(\z\) satisfying
\[
  p(\z)\ge\OPT-2\Phi_\varepsilon(\widehat{\x}).
\]
\end{theorem}

\begin{proof}
By \cref{lemma:deterministic-rounding-multilinear}, \cref{eqn:loss-decomposition}, and \(g_b(\x^*)\ge0\),
\[
  \OPT-p(\z_b)
  \le
  \OPT-p(\y_b)
  \le
  \zeta(\x^*)-\zeta(\y_b).
\]
The bounds \(\zeta(\x^*)\le\Phi_\varepsilon(\widehat{\x})\) and \(-\zeta(\y_b)\le\Phi_b(\widehat{\x})\) of \cref{lemma:footprint-control} give the first claim.
By the same lemma, \(\x^*\in P_\varepsilon\), so \(\OPT-p(\z_\varepsilon)\le2\Phi_\varepsilon(\widehat{\x})\), and the full scan returns a candidate at least as good as \(\z_\varepsilon\).
\end{proof}
\paragraph{Enumerating the error budget.}
The proof of \cref{thm:error-footprint-guarantee} uses only one budget that covers the true error, and a budget within a factor of two of the error suffices.
The full scan in \cref{algorithm: pipeline} can therefore be replaced by the doubling grid
\[
  \mathcal B=\{0,n\}\cup\{2^j:j\in\mathbb Z_{\ge0},\ 2^j<n\},
\]
which has \(O(\log n)\) elements.

\begin{proposition}[Doubling budget grid]\label{prop:doubling-grid}
Under the conditions of \cref{thm:error-footprint-guarantee}, \LEAP{} run over the budgets in \(\mathcal B\) solves \(O(\log n)\) linear programs and returns an assignment \(\z\) satisfying
\[
  p(\z)\ge\OPT-3\Phi_\varepsilon(\widehat{\x}).
\]
\end{proposition}

The proof, given in \cref{app:footprint-guarantee}, uses a grid budget \(b\) with \(\varepsilon\le b<2\varepsilon\), for which \(\Phi_b(\widehat{\x})\le2\Phi_\varepsilon(\widehat{\x})\).

We state the remaining results for the full scan; each extends to the doubling grid with \(3\Phi_\varepsilon(\widehat{\x})\) in place of \(2\Phi_\varepsilon(\widehat{\x})\), and for \cref{thm:lp-refined-guarantee} with \(\kappa_\varepsilon^*\) replaced by \(\kappa_b^*\) at the grid budget \(b\) used in the proof of \cref{prop:doubling-grid}.

\paragraph{Complexity.}
For fixed degree \(d\), the decomposition contributes \(O(n^{d-1})\) tuple-indexed constraints and polynomially many variables in the worst case, so each relaxation has polynomial description size.
The relaxation is much smaller for sparse objectives.
Under the decomposition constructed in the proof of \cref{lemma:decomposition-of-polynomials}, a constrained node lies on the path of some monomial of \(p\), and each monomial passes through at most \(d-1\) constrained nodes.
The relaxation therefore has at most \((d-1)M_p\) coefficient constraints, where \(M_p\) is the number of monomials of \(p\); for \MAXKSAT{} with clause set \(\mathcal C\), \(M_p\le2^k|\mathcal C|\).
Let \(N_d\) denote this description size and \(T_{\mathtt{LP}}\) the running time of the chosen polynomial-time LP solver.
Since the LP solves dominate the running time, full budget enumeration takes \(O(n\cdot T_{\mathtt{LP}}(N_d))\) time, and the doubling grid of \cref{prop:doubling-grid} reduces this to \(O(\log n\cdot T_{\mathtt{LP}}(N_d))\).
With a perfect prediction, the error footprint is zero, and for multilinear objectives deterministic rounding returns an optimal solution.

\subsection{Refinement via LP Duality}\label{sec:lp-refinement}

The surrogate gain \(g_b(\x^*)\) in \cref{eqn:loss-decomposition} is positive whenever \(\x^*\) does not maximize the surrogate over \(P_b\).
We characterize this gain through LP dual prices and use it to tighten \cref{thm:error-footprint-guarantee}.

Write \cref{optimization:d-LP-oracle} at budget \(b\) as
\[
  \max_{\mathbf w}\ \bigl\{c_{\emptyset}+\bar{\mathbf a}^\top\mathbf w:\ G\mathbf w\le\mathbf h\bigr\},
  \qquad
  \mathbf w=(\y,\mathbf u),
\]
where the rows of \(G\mathbf w\le\mathbf h\) are the coefficient constraints, the budget constraint, and the bounds on \(\y\) and \(\mathbf u\), and \(\bar{\mathbf a}\) has zero entries on \(\mathbf u\).
For a feasible \(\mathbf w\), let \(\boldsymbol\sigma(\mathbf w)=\mathbf h-G\mathbf w\ge0\) be its vector of constraint slacks.

\begin{lemma}[LP gain]\label{lemma:lp-gain}
Let \(\boldsymbol\lambda\ge0\) be an optimal dual solution of \cref{optimization:d-LP-oracle} at budget \(b\).
Every feasible \(\mathbf w=(\x,\mathbf u)\) satisfies
\[
  g_b(\x)=\boldsymbol\lambda^\top\boldsymbol\sigma(\mathbf w).
\]
\end{lemma}

\begin{proof}
Dual feasibility gives \(G^\top\boldsymbol\lambda=\bar{\mathbf a}\), and strong duality gives \(Q_b=c_{\emptyset}+\mathbf h^\top\boldsymbol\lambda\).
Hence \(g_b(\x)=Q_b-c_{\emptyset}-\bar{\mathbf a}^\top\mathbf w=\boldsymbol\lambda^\top(\mathbf h-G\mathbf w)\).
\end{proof}

Intuitively, \cref{lemma:lp-gain} shows that the surrogate gain prices each constraint slack by its dual multiplier.
To quantify what fraction of the linearization discrepancy \(\zeta(\x)-\zeta(\y_b)\) is offset by this gain, we define
\begin{equation}\label{eqn:contraction-coefficient}
  \kappa_b^*
  =
  \max\Bigl\{\kappa\in[0,1]:\ g_b(\x)\ge\kappa\bigl(\zeta(\x)-\zeta(\y_b)\bigr)\ \text{for every }\x\in P_b\Bigr\}.
\end{equation}
The maximum exists, since each \(\x\) imposes a closed condition and \(\kappa=0\) always qualifies.
By \cref{lemma:lp-gain}, \(\kappa_b^*\) measures how much of the change in linearization error the dual prices of the slacks can pay for: \(\kappa_b^*=0\) recovers the footprint bound of \cref{thm:error-footprint-guarantee}, and \(\kappa_b^*=1\) means that \(\y_b\) maximizes \(p\) over \(P_b\).

\begin{theorem}[LP-refined guarantee]\label{thm:lp-refined-guarantee}
Under the conditions of \cref{thm:error-footprint-guarantee}, for every budget \(b\) with \(\x^*\in P_b\),
\[
  \OPT-p(\z_b)
  \le
  (1-\kappa_b^*)\bigl[\Phi_\varepsilon(\widehat{\x})-\zeta(\y_b)\bigr]
  \le
  (1-\kappa_b^*)\bigl[\Phi_\varepsilon(\widehat{\x})+\Phi_b(\widehat{\x})\bigr].
\]
Consequently, \LEAP{} with deterministic rounding returns an integral assignment \(\z\) satisfying
\[
  p(\z)\ge\OPT-2(1-\kappa_\varepsilon^*)\Phi_\varepsilon(\widehat{\x}).
\]
\end{theorem}

\begin{proof}
By \cref{lemma:deterministic-rounding-multilinear}, \cref{eqn:loss-decomposition}, and \cref{eqn:contraction-coefficient} at \(\x=\x^*\),
\[
  \OPT-p(\z_b)
  \le
  \OPT-p(\y_b)
  =
  \bigl[\zeta(\x^*)-\zeta(\y_b)\bigr]-g_b(\x^*)
  \le
  (1-\kappa_b^*)\bigl[\zeta(\x^*)-\zeta(\y_b)\bigr].
\]
Since \(1-\kappa_b^*\ge0\), the bounds \(\zeta(\x^*)\le\Phi_\varepsilon(\widehat{\x})\) and \(-\zeta(\y_b)\le\Phi_b(\widehat{\x})\) of \cref{lemma:footprint-control} give the two inequalities.
By the same lemma, \(\x^*\in P_\varepsilon\), so \(\OPT-p(\z_\varepsilon)\le2(1-\kappa_\varepsilon^*)\Phi_\varepsilon(\widehat{\x})\), and the full scan returns a candidate at least as good as \(\z_\varepsilon\).
\end{proof}

Since \(\kappa_b^*\ge0\), \cref{thm:lp-refined-guarantee} is never weaker than \cref{thm:error-footprint-guarantee}, and if \(\kappa_b^*=1\) at a budget with \(\x^*\in P_b\), then \(\z_b\) is optimal.
\Cref{app:positive-contraction} characterizes when \(\kappa_b^*\) is positive and shows how lower bounds on \(\kappa_b^*\) can be certified from an optimal dual solution.

\subsection{Analyzing the Error Footprint}\label{sec:footprint-budget-analysis}
\Cref{thm:lp-refined-guarantee} sharpens the guarantee once the LP is solved, while \cref{thm:error-footprint-guarantee} gives an \emph{a priori} guarantee as soon as the error footprint is bounded.
We bound \(\Phi_\varepsilon(\widehat{\x})\) first globally, using only smoothness (\cref{def:smooth-polynomial}), and then locally, through the order of the \(\sigma\)-decomposition.

\begin{lemma}[Global footprint bound]\label{lemma:global-footprint-bound}
Let \(p(\x)\) be an \(n\)-variate degree-\(d\), \(\beta\)-smooth polynomial with \(d\ge2\) and \(n>d\).
There is a constant \(\rho_d\) depending only on \(d\) such that
\[
  \Phi_\varepsilon(\widehat{\x})
  \le
  \rho_d\beta n^{d-1}\varepsilon .
\]
\end{lemma}

The proof, given in \cref{app:global-footprint}, bounds the contribution of each level of the decomposition tree by the magnitude bound of \citet[Lemma~3.2]{Arora1999PTAS}.

\begin{corollary}[Global footprint consequence]\label{cor:global-smoothness-guarantee}
Let \(p\) be a multilinear degree-\(d\), \(\beta\)-smooth objective with \(d\ge2\) and \(n>d\), let \(\x^*\) be an optimal assignment, and set \(\varepsilon=\|\widehat{\x}-\x^*\|_1\).
Under deterministic rounding, \LEAP{} returns an assignment \(\z\) satisfying
\[
  p(\z)
  \ge
  \OPT-2\rho_d\beta n^{d-1}\varepsilon .
\]
\end{corollary}

\begin{proof}
Combine \cref{thm:error-footprint-guarantee} with \cref{lemma:global-footprint-bound}.
\end{proof}

The global bound ignores which variables appear in each component.
To capture local structure, write the multilinear objective as \(p(\x)=\sum_{M\subseteq[n]}a_M\x_M\) with \(\x_M=\prod_{i\in M}x_i\), fix an order \(\sigma\) of the variables, and build \cref{optimization:d-LP-oracle} from the \(\sigma\)-decomposition.
The \emph{ordered load} of coordinate \(j\) is
\[
  L_j^\sigma
  =
  \sum_{M\ni j:\ j\ne\min_\sigma M}|a_M|,
\]
the total weight of the monomials that contain \(j\) but do not start with \(j\) in the order \(\sigma\).
Let \(L^\sigma_{(1)}\ge L^\sigma_{(2)}\ge\cdots\) denote the ordered loads in decreasing order.

\begin{lemma}[Ordered-load footprint bound]\label{lemma:coordinate-load-footprint-bound}
Under the \(\sigma\)-decomposition, every prediction \(\widehat{\x}\in\{0,1\}^n\) satisfies \(F_j(\widehat{\x})\le L_j^\sigma\) for every \(j\in[n]\), and hence, for every integer budget \(b\),
\[
  \Phi_b(\widehat{\x})
  \le
  \sum_{t=1}^{b}L^\sigma_{(t)}
  \le
  b\max_{j\in[n]}L_j^\sigma .
\]
\end{lemma}

The proof, given in \cref{app:global-footprint}, observes that each monomial charges every variable other than its \(\sigma\)-first one exactly once, with magnitude at most \(|a_M|\) because \(\widehat{\x}\) is Boolean.

Combined with \cref{thm:error-footprint-guarantee}, the lemma bounds the loss of \LEAP{} by \(2\sum_{t=1}^{\varepsilon}L^\sigma_{(t)}\).
The order \(\sigma\) also decides which errors are charged at all, since a mispredicted variable that comes first in every monomial containing it moves no local coefficient; this gives the following sharper bound.

\begin{proposition}[Structural gain]\label{prop:structural-gain}
Let \(p\) be a multilinear objective, let \(\x^*\) be an optimal assignment, and build \LEAP{} from the \(\sigma\)-decomposition.
Let \(S=\{j:\widehat{x}_j\ne x^*_j\}\), and let
\(
  \varepsilon_\sigma=|\{j\in S:\ L_j^\sigma>0\}|\le\varepsilon
\)
count the mispredicted variables with positive ordered load.
Then \LEAP{} returns an assignment \(\z\) satisfying
\[
  p(\z)
  \ge
  \OPT-2\sum_{t=1}^{\varepsilon_\sigma}L^\sigma_{(t)}
  \ge
  \OPT-2\varepsilon_\sigma\max_{j\in[n]}L_j^\sigma .
\]
In particular, if \(\varepsilon_\sigma=0\), then \LEAP{} returns an optimal assignment.
\end{proposition}

The proof, given in \cref{app:structural-gain}, shows that \(\x^*\) is already feasible at a budget of at most \(\varepsilon_\sigma\); since \LEAP{} scans all budgets, this smaller budget enters only the analysis.

\Cref{prop:structural-gain} reveals how \LEAP{} exploits the structure of the problem.
A prediction error affects the outcome only through the local coefficients it distorts, and the decomposition order decides which coefficients each variable enters.
\Cref{sec:app-maxcut} gives a simple example on \MAXCUT{}.

\section{Applications}\label{sec:applications}

We instantiate the framework on two examples.
The first is \MAXCUT{}, a quadratic graph cut problem, and the second is \MAXKSAT{}, a Boolean constraint satisfaction problem with higher-order interactions.
For each problem, we derive a global bound and a local bound, and the local bound depends on the degeneracy of the instance, which governs the best decomposition order.
We then show that the local bounds obtained for both problems have optimal order.

\subsection{\MAXCUT{}}\label{sec:app-maxcut}
The \MAXCUT{} objective of \cref{example: maxcut} is quadratic, multilinear, and \(2\)-smooth.

\paragraph{Global bound.}
\Cref{cor:global-smoothness-guarantee} gives \(p(\z)\ge\OPT-O(n\varepsilon)\), so the approximation ratio is \(1-O(\varepsilon/n^\gamma)\) whenever \(\OPT=\Omega(n^{1+\gamma})\) for a constant \(\gamma\in[0,1]\), in particular whenever the average degree is \(\Omega(n^\gamma)\).

\paragraph{Local bound.}
Under the \(\sigma\)-decomposition, \(L_j^\sigma\) is twice the number of neighbors of \(j\) that precede it in \(\sigma\).
The \emph{degeneracy} \(\dgn(G)\) is the smallest \(r\) for which some order gives every vertex at most \(r\) earlier neighbors, and such an order is computed in linear time~\citep{Matula1983SmallestLast}.
Under this order, \(\max_jL_j^\sigma\le2\dgn(G)\), so \cref{prop:structural-gain} gives
\[
  p(\z)\ge\OPT-4\dgn(G)\,\varepsilon_\sigma\ge\OPT-4\dgn(G)\,\varepsilon\ge\OPT-4\Delta\varepsilon,
\]
where \(\Delta\) is the maximum degree, and the approximation ratio is \(1-O(\dgn(G)\,\varepsilon/n)\) on graphs without isolated vertices.
The degeneracy is often much smaller than \(\Delta\): it is one on trees, at most five on planar graphs, and at most \(m\) on Barab\'asi--Albert graphs in which every new vertex attaches \(m\) edges.

By \cref{prop:structural-gain}, errors at vertices placed before all of their neighbors are not charged at all; the following example consists of such errors only.

\begin{example}[Exact recovery on stars]\label{example:stars}
Let \(G\) be the disjoint union of \(\varepsilon\) stars \(K_{1,\Delta}\) with \(\Delta\ge2\), and let the oracle predict the all-ones assignment.
The prediction errs only at the \(\varepsilon\) centers and cuts no edge, while \(\OPT=\Delta\varepsilon\).
Ordering every center before its leaves gives every center zero ordered load, so \(\varepsilon_\sigma=0\) and \LEAP{} returns an optimal cut.
\end{example}

\subsection{\MAXKSAT{}}
The \MAXKSAT{} objective of \cref{example: maxksat} is multilinear of degree \(k\) and \(4^k\)-smooth.

\paragraph{Global bound.}
\Cref{cor:global-smoothness-guarantee} gives \(p(\z)\ge\OPT-O(n^{k-1}\varepsilon)\), so the approximation ratio is \(1-O(\varepsilon/n^\gamma)\) whenever \(\OPT=\Omega(n^{k-1+\gamma})\).

\paragraph{Local bound.}
For an order \(\sigma\), let \(\operatorname{occ}_\sigma(j)\) count the clauses that contain \(j\) but do not start with \(j\) in \(\sigma\), and define the degeneracy of the formula as \(\dgn(\phi)=\min_\sigma\max_j\operatorname{occ}_\sigma(j)\).
As for graphs, a smallest-last rule computes an order attaining \(\dgn(\phi)\) in polynomial time.

\begin{lemma}\label{lemma:maxsat-footprint}
For every order \(\sigma\), \(L_j^\sigma\le2^{k-1}\operatorname{occ}_\sigma(j)\) for every \(j\).
If each literal occurs in at most \(\Delta\) clauses, then \(\dgn(\phi)\le2\Delta\).
\end{lemma}

\begin{proof}
Each clause expands into monomials with coefficients \(\pm1\), at most \(2^{k-1}\) of which contain a given variable \(j\), and all of them start with \(j\) when \(j\) comes first in the clause.
The second claim holds because \(j\) lies in at most \(2\Delta\) clauses.
\end{proof}

Under an order attaining \(\dgn(\phi)\), \cref{prop:structural-gain} gives
\[
  p(\z)\ge\OPT-2^k\dgn(\phi)\,\varepsilon_\sigma\ge\OPT-2^k\dgn(\phi)\,\varepsilon\ge\OPT-2^{k+1}\Delta\varepsilon,
\]
and since \(\OPT\ge n/(2k)\) once variables that appear in no clause are removed, the approximation ratio is \(1-O(\dgn(\phi)\,\varepsilon/n)\).

\subsection{Optimality of the Local Bounds}\label{sec:local-order-optimality}
The local bounds above grow linearly in \(\dgn\cdot\varepsilon\), and \LEAP{} attains them without knowing \(\varepsilon\).
Unless \(\P=\NP\), no polynomial-time algorithm can guarantee a smaller order, already for \MAXCUT{} and even when the error \(\varepsilon\) with respect to a nearest optimal assignment is given as part of the input.

\begin{proposition}[Worst-case order of the local bound]\label{prop:local-order-lower-bound}
Assume \(\P\ne\NP\), and consider deterministic polynomial-time algorithms that receive a \MAXCUT{} instance, a prediction \(\widehat{\x}\), and its error \(\varepsilon\).
\begin{enumerate}[label=(\roman*), leftmargin=*]
  \item On instances of maximum degree at most \(4\), no such algorithm guarantees \(\OPT-p(\z)\le\psi(\varepsilon)\,\varepsilon\) for any function \(\psi\) with \(\psi(\varepsilon)\to0\) as \(\varepsilon\to\infty\).
  \item For \(\theta\in\{\dgn(G),\Delta\}\), no such algorithm guarantees \(\OPT-p(\z)\le\psi(\theta)\,\theta\varepsilon\) for any function \(\psi\) with \(\psi(\theta)\to0\) as \(\theta\to\infty\).
\end{enumerate}
Both statements also hold for \MAXKSAT{} for every fixed \(k\ge2\), with \(\theta\in\{\dgn(\phi),\Delta\}\), where \(\Delta\) is the maximum number of occurrences of a literal, and with the degree bound in~(i) replaced by a constant depending only on \(k\).
\end{proposition}

The proof in \cref{app:local-order} reduces from the APX-hardness of \MAXCUT{} on graphs of maximum degree three~\citep{Alimonti2000Cubic}.\footnote{For randomized algorithms that meet their guarantee with probability at least \(2/3\), the same reduction applies under the assumption \(\NP\not\subseteq\mathrm{BPP}\).}
Attaching a private leaf to every vertex fixes the error of the all-zero prediction, and replacing every vertex by \(t\) copies raises both the degeneracy and the maximum degree to \(\Theta(t)\).
Adding correctly predicted disjoint edges also rules out approximation ratios \(1-h(\theta\varepsilon/n)\) with \(h(s)=o(s)\), so the ratios \(1-O(\dgn\cdot\varepsilon/n)\) derived above have optimal order.

Together, \cref{sec:approach,sec:applications} answer \hyperref[q:approximation-from-prediction]{Q1}: given any prediction, \LEAP{} runs in polynomial time with an error-footprint guarantee, and on instances of bounded degeneracy no polynomial-time algorithm improves the order of this guarantee.
We next turn to \hyperref[q:learned-oracle-performance]{Q2} and ask how to learn a good predictor from sampled optimization instances.
\section{Learnability of Oracles}\label{sec:learnability}

This section studies how to select a predictor from sampled optimization instances.
Following the PAC perspective of \citet{Gupta2017PACApproach}, we view each predictor \(f\) as inducing an algorithm \(\mathcal A_f\): its prediction is fed into the footprint relaxation and the subsequent rounding procedure.

\subsection{Oracle-Induced Algorithm Classes}

Let \(\Pi\) denote the instance space.
Each instance \(\pi\in\Pi\) is a smooth Boolean polynomial optimization problem of the form \cref{eqn:d-IP-intro}; write \(p_\pi\) for its objective.
We assume, for simplicity of exposition, that all instances in the family have the same fixed size (\cref{sec:pseudo-dimension} discusses instances of different sizes) and that all objective values lie in a common range \([0,H]\) with \(H>0\).
Let \(\mathcal D\) be an unknown but fixed distribution over \(\Pi\).

Let \(\mathcal F\) be a hypothesis class of predictors \(f:\Pi\to\{0,1\}^n\), viewed as candidate oracles.
For \(f\in\mathcal F\) and instance \(\pi\), \LEAP{} uses \(\widehat{\x}=f(\pi)\) and returns an integral assignment \(\z_f(\pi)\); unless stated otherwise, \LEAP{} refers to \cref{algorithm: pipeline} with the full budget scan.
Thus \(f\) induces a \LEAP{} algorithm \(\mathcal A_f\), and we write
\[
  \mathcal A_{\mathcal F}=\{\mathcal A_f:f\in\mathcal F\}
\]
for the induced algorithm class.

Let \(\x^*(\pi)\) denote an optimal assignment for instance \(\pi\).
The expected value achieved by predictor \(f\) is
\[
  \mathcal P(f)
  \triangleq
  \Exp_{\pi\sim\mathcal D}\bigl[p_\pi(\z_f(\pi))\bigr],
\]
and the expected optimum is
\[
  \mathcal P^*
  \triangleq
  \Exp_{\pi\sim\mathcal D}\bigl[p_\pi(\x^*(\pi))\bigr].
\]

\subsection{Empirical Cost Minimization}

We select predictors according to the performance of the algorithms they induce.
Since objective values lie in \([0,H]\), define the cost of \(\mathcal A_f\) on \(\pi\) as
\[
  \mathtt{COST}(\mathcal A_f,\pi)
  \triangleq
  H-p_\pi(\z_f(\pi)).
\]
For example, in a \MAXKSAT{} family with a fixed number of clauses, \(H\) may be taken as the number of clauses, and the cost is the number of clauses not satisfied by the returned assignment.

The expected risk is
\[
  R(f)
  \triangleq
  \Exp_{\pi\sim\mathcal D}\bigl[\mathtt{COST}(\mathcal A_f,\pi)\bigr],
\]
and let
\[
  f^*_{\mathtt{cost}}\in\argmin_{f\in\mathcal F}R(f)
\]
be a cost-optimal oracle in the hypothesis class.
Since \(\mathtt{COST}(\mathcal A_f,\pi)=H-p_\pi(\z_f(\pi))\), we have \(R(f)=H-\mathcal P(f)\).

Given an i.i.d. sample \(S=\{\pi_1,\dots,\pi_m\}\) from \(\mathcal D\), empirical risk minimization (ERM) selects
\[
  \hat f_{\mathrm{ERM}}
  \in
  \argmin_{f\in\mathcal F}
  \frac1m\sum_{i=1}^m
  \mathtt{COST}(\mathcal A_f,\pi_i).
\]

\subsection{Pseudo-Dimension Bound}\label{sec:pseudo-dimension}

We bound the complexity of the induced cost class.
Since the cost is real-valued, the relevant notion is pseudo-dimension, the standard real-valued analogue of VC dimension~\citep{pollard1990empirical, anthony1999neural}.
For each coordinate \(i\in[n]\), let \(\mathcal F_i\) be the class of \(i\)-th coordinate predictors, and define
\[
  d_{\mathcal F}
  \triangleq
  \sum_{i=1}^n \mathrm{VCdim}(\mathcal F_i).
\]
We assume \(d_{\mathcal F}\ge1\).

\begin{restatable}{lemma}{RestateLemmaPseudoDimension}\label{lemma:pseudo-dimension-of-algorithm-class}
Let \(\mathcal F\) be a hypothesis class of predictors \(f:\Pi\to\{0,1\}^n\), and let \(\mathcal A_{\mathcal F}\) be the induced class of \LEAP{} algorithms.
The pseudo-dimension of the cost functions
\[
  \bigl\{\pi\mapsto \mathtt{COST}(\mathcal A_f,\pi):f\in\mathcal F\bigr\}
\]
satisfies
\[
  \mathrm{Pdim}(\mathcal A_{\mathcal F})
  \le
  C\,d_{\mathcal F}\log(e\,d_{\mathcal F})
\]
for some absolute constant \(C\).
\end{restatable}

\begin{proof}
Throughout the proof, we fix deterministic tie-breaking rules for the LP solver, the rounding procedure, and the final candidate selection.
If the algorithm uses internal randomness, we include the random seed in the input domain.
Let \(\Xi\) denote the seed space, and let \(\z_f(\pi;\xi)\) be the assignment returned on \(\pi\) with seed \(\xi\). Define
\[
  h_f:\Pi\times\Xi\to[0,H],
  \qquad
  h_f(\pi,\xi)=H-p_\pi\bigl(\z_f(\pi;\xi)\bigr).
\]
It suffices to bound the pseudo-dimension of \(\mathcal H=\{h_f:f\in\mathcal F\}\).

Fix a finite set \(S=\{(\pi_1,\xi_1),\dots,(\pi_m,\xi_m)\}\subseteq\Pi\times\Xi\) and thresholds \(r_1,\dots,r_m\in\R\).
For each \(f\in\mathcal F\), define a labeling vector
\[
  \ell_f(j)=\mathbf 1[h_f(\pi_j,\xi_j)>r_j].
\]
For fixed \((\pi_j,\xi_j)\), the final cost is determined by the prediction \(f(\pi_j)\), since the instance and seed fix the subsequent relaxation and rounding outcomes.
Thus the number of possible labelings \(\{\ell_f\}\) is at most the number of distinct prediction matrices
\[
  P_f\in\{0,1\}^{n\times m},
  \qquad
  [P_f]_{i,j}=f_i(\pi_j).
\]

Let \(d_i=\mathrm{VCdim}(\mathcal F_i)\), and assume \(m\ge1\).
If \(d_i=0\), the coordinate class yields at most one distinct vector on the sample.
For \(1\le d_i\le m\), Sauer's lemma~\citep{Sauer1972lemma} bounds the number of distinct vectors
\((f_i(\pi_1),\dots,f_i(\pi_m))\) by
\[
  \sum_{s=0}^{d_i}\binom{m}{s}
  \le
  \left(\frac{em}{d_i}\right)^{d_i}
  \le (em)^{d_i}.
\]
If \(d_i>m\), the trivial bound \(2^m\le(em)^{d_i}\) applies.
Thus \((em)^{d_i}\) bounds the number of coordinate vectors in all cases.
Even when the coordinates of \(f\) are coupled, the set of valid prediction matrices is a subset of the Cartesian product of the coordinate-wise projections.
Hence the number of distinct prediction matrices is at most
\[
  \prod_{i=1}^n(em)^{d_i}
  =
  (em)^{d_{\mathcal F}}.
\]
If \(S\) is pseudo-shattered by \(\mathcal H\), then all \(2^m\) binary labelings must be realized, so
\[
  2^m\le (em)^{d_{\mathcal F}}.
\]
This inequality fails for \(m>C\,d_{\mathcal F}\log(e d_{\mathcal F})\) with a sufficiently large absolute constant \(C\), which proves the bound.
\end{proof}

\paragraph{Oracle families.}
The bound applies to predictor classes with controlled coordinate-wise VC dimension.
This includes bounded-size GNN- and Transformer-based predictors for graph and formula instances~\citep{Selsam2019LBLMD, Gasse2019ML4co, pan2025can}, which have finite VC dimension~\citep{bartlett2019nearly}.
Our analysis fixes the instance size \(n\).
Parameter-sharing predictors such as GNNs and parameterized local search are also defined across instance sizes, and \cref{app:empirical-selection-details} reports selection experiments that transfer candidates from smaller to larger instances.
The computational efficiency of the ERM step is a separate question and remains open at this level of generality~\citep{Gupta2017PACApproach, bartlett2022generalization}.

\subsection{End-to-End Learnability Guarantee}

The pseudo-dimension bound gives a standard ERM guarantee for the induced algorithm class.
We combine this guarantee with the footprint bound of \cref{thm:error-footprint-guarantee} through the following predictor-class quantity.
For an instance \(\pi\), let \(\Phi^\pi_\varepsilon(\widehat{\x})\) denote the error footprint induced by \(p_\pi\), \(\widehat{\x}\), and \(\varepsilon\), and define
\[
  \Lambda_{\mathcal F}^{\Phi}
  \triangleq
  2\inf_{f\in\mathcal F}
  \Exp_{\pi\sim\mathcal D}
  \left[
    \Phi^\pi_{\|f(\pi)-\x^*(\pi)\|_1}(f(\pi))
  \right].
\]
This quantity is the footprint loss of the best predictor in \(\mathcal F\).
The next lemma is where the footprint guarantee enters the learnability analysis.

\begin{lemma}[Transfer lemma]\label{lemma:cost-optimal-transfer}
For the \LEAP{} pipeline of \cref{algorithm: pipeline},
\[
  \mathcal P(f^*_{\mathtt{cost}})
  \ge
  \mathcal P^*
  -
  \Lambda_{\mathcal F}^{\Phi}.
\]
\end{lemma}

\begin{proof}
By the definition of \(f^*_{\mathtt{cost}}\), for every \(f\in\mathcal F\),
\[
  R(f^*_{\mathtt{cost}})
  \le
  R(f).
\]
Since \(R(f)=H-\mathcal P(f)\), this is equivalent to
\[
  \mathcal P(f^*_{\mathtt{cost}})
  \ge
  \mathcal P(f).
\]
For every \(f\in\mathcal F\) and instance \(\pi\), \cref{thm:error-footprint-guarantee} gives
\[
  p_\pi(\z_f(\pi))
  \ge
  p_\pi(\x^*(\pi))
  -
  2\Phi^\pi_{\|f(\pi)-\x^*(\pi)\|_1}(f(\pi)).
\]
Taking expectations, combining with \(\mathcal P(f^*_{\mathtt{cost}})\ge \mathcal P(f)\), and taking the infimum over \(f\in\mathcal F\) yields
\[
  \mathcal P(f^*_{\mathtt{cost}})
  \ge
  \mathcal P^*
  -
  \Lambda_{\mathcal F}^{\Phi}.
\]
\end{proof}

Combining the ERM excess-risk bound with the transfer lemma gives the final end-to-end guarantee.

\begin{restatable}{theorem}{RestateThmErmLearnability}\label{thm:erm-learnability}
Under the setup above, fix \(\alpha>0\) and \(\delta\in(0,1]\).
If
\[
  m\ge
  C\left(\frac{H}{\alpha}\right)^2
  \left(d_{\mathcal F}\log(e d_{\mathcal F})+\log\frac1\delta\right),
\]
then with probability at least \(1-\delta\), any empirical risk minimizer \(\hat f_{\mathrm{ERM}}\) satisfies \(R(\hat f_{\mathrm{ERM}})-R(f^*_{\mathtt{cost}})\le 2\alpha\), and consequently
\[
  \mathcal P(\hat f_{\mathrm{ERM}})
  \ge
  \mathcal P^*
  -
  \Lambda_{\mathcal F}^{\Phi}
  -
  2\alpha.
\]
\end{restatable}

\begin{proof}
By \cref{lemma:pseudo-dimension-of-algorithm-class}, the induced cost class has pseudo-dimension \(O(d_{\mathcal F}\log(e d_{\mathcal F}))\).
Applying the standard uniform-convergence bound for ERM (\cref{coro:erm-learning} in \cref{sec:omitted-proofs-learnability}) with this bound gives the excess-risk statement after adjusting the absolute constant \(C\).
Since \(R(f)=H-\mathcal P(f)\), the excess-risk bound is equivalent to
\[
  \mathcal P(\hat f_{\mathrm{ERM}})
  \ge
  \mathcal P(f^*_{\mathtt{cost}})-2\alpha.
\]
Combining this inequality with \cref{lemma:cost-optimal-transfer} gives the expected optimization guarantee.
\end{proof}

The footprint bounds in \cref{sec:approach,sec:applications} turn \(\Lambda_{\mathcal F}^{\Phi}\) into concrete bounds involving the best expected prediction error
\[
  \varepsilon_{\mathcal F}
  \triangleq
  \inf_{f\in\mathcal F}
  \Exp_{\pi\sim\mathcal D}
  \bigl[\|f(\pi)-\x^*(\pi)\|_1\bigr].
\]
The global smoothness bound in \cref{cor:global-smoothness-guarantee} gives \(\Lambda_{\mathcal F}^{\Phi}=O(\beta n^{d-1}\varepsilon_{\mathcal F})\).
For \MAXKSAT{}, if every formula in the support of \(\mathcal D\) has degeneracy at most \(r\), which holds with \(r=2\Delta\) when each literal occurs in at most \(\Delta\) clauses, and \LEAP{} uses an order attaining it, then \cref{lemma:maxsat-footprint,prop:structural-gain} give \(\Lambda_{\mathcal F}^{\Phi}=O(r\,\varepsilon_{\mathcal F})\); the order depends only on the instance, so \cref{lemma:pseudo-dimension-of-algorithm-class} still applies.

\paragraph{LP refinement.}
Using \cref{thm:lp-refined-guarantee} in place of \cref{thm:error-footprint-guarantee} in the proof of \cref{lemma:cost-optimal-transfer} multiplies the error footprint inside the expectation defining \(\Lambda_{\mathcal F}^{\Phi}\) by \(1-\kappa^*_{\|f(\pi)-\x^*(\pi)\|_1}\), the coefficient of \cref{eqn:contraction-coefficient} for instance \(\pi\) under prediction \(f(\pi)\) with the tie-breaking rules fixed in the proof of \cref{lemma:pseudo-dimension-of-algorithm-class}; any measurable lower bound on this coefficient may be used instead.
The resulting loss term is never larger than \(\Lambda_{\mathcal F}^{\Phi}\), and ERM is unchanged, since the coefficients enter only the analysis.

\paragraph{Doubling grid.}
With the doubling grid of \cref{prop:doubling-grid}, \cref{lemma:cost-optimal-transfer,thm:erm-learnability} hold with \(\Lambda_{\mathcal F}^{\Phi}\) replaced by \(\tfrac32\Lambda_{\mathcal F}^{\Phi}\), so the concrete bounds above keep their order.
The pseudo-dimension bound uses only that the cost on a fixed instance is determined by the prediction, and \cref{prop:doubling-grid} replaces \(2\Phi\) by \(3\Phi\) in the transfer lemma.
\section{Experiments}\label{sec:experiments}

The experiments cover the two levels of the framework, approximation from a given prediction and oracle selection from sampled instances, and then test \LEAP{} at scale.
\Cref{sec:solution-quality} controls the prediction error exactly and tests the global and local footprint bounds of \cref{sec:approach,sec:applications}.
\Cref{sec:empirical-selection-oracle-candidates} tests oracle selection by empirical induced cost, as analyzed in \cref{sec:learnability}.
\Cref{sec:sota-maxsat} runs \LEAP{} on large \MAXSAT{} application benchmarks.

\subsection{Solution Quality from Imperfect Predictions}\label{sec:solution-quality}

To control the prediction error exactly, we use \emph{planted} instances, which are generated so that a fixed assignment \(\x^*\) cuts every edge for \MAXCUT{} and satisfies every clause for \MAXSAT{}, and is therefore optimal.
The prediction flips a prescribed fraction of the variables in \(\x^*\), so the prediction error equals the number of flipped variables.
We compare two uses of this prediction: the direct baseline returns the predicted assignment itself, while \LEAP{} passes the prediction through the relaxation-and-rounding pipeline with the full budget scan of \cref{algorithm: pipeline}.
Both methods are scored by the relative objective gap \((\OPT-p(\z))/\OPT\) for \(n\in\{100,200,500,1000\}\) and flip fractions from \(0\%\) to \(20\%\); construction details are given in \cref{app:controlled-perturbations}.

\paragraph{Dense instances.}
This regime tests the global smoothness bound of \cref{sec:applications}: with \(\OPT=\Theta(n^d)\), it predicts an approximation ratio of \(1-O(\varepsilon/n)\), where \(\varepsilon/n\) is the fraction of flipped variables.
\Cref{fig:near-dense-controlled} shows that the gap of the direct prediction grows roughly linearly with the flip fraction, while \LEAP{} recovers the planted optimum at every size and flip fraction.

\begin{figure}[!htbp]
  \centering
  \includegraphics[width=\linewidth]{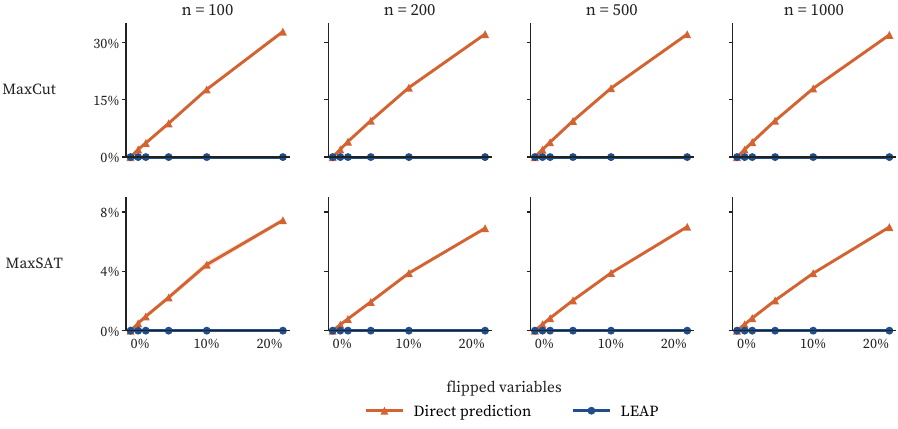}
  \caption{
  Relative objective gaps versus the fraction of flipped variables for \emph{dense} instances.
  The gap of the direct prediction grows roughly linearly with the flip fraction, while \LEAP{} repairs the perturbed predictions and recovers the planted optimum at every size.}
  \label{fig:near-dense-controlled}
\end{figure}

\paragraph{Sparse instances.}
This regime tests the local ordered-load bound of \cref{prop:structural-gain} on \MAXCUT{} and bounded-occurrence \MAXSAT{} instances.
Under the natural variable order used here, each ordered load is at most a constant multiple of the degree or occurrence count, so the bound charges each error only to the edges or clauses it affects.
\Cref{fig:large-flip-repair} shows that \LEAP{} reduces the gap of the direct prediction at every positive flip fraction, including the largest one; the largest residual gaps occur on the larger \MAXCUT{} instances.

\begin{figure}[!htbp]
  \centering
  \includegraphics[width=\linewidth]{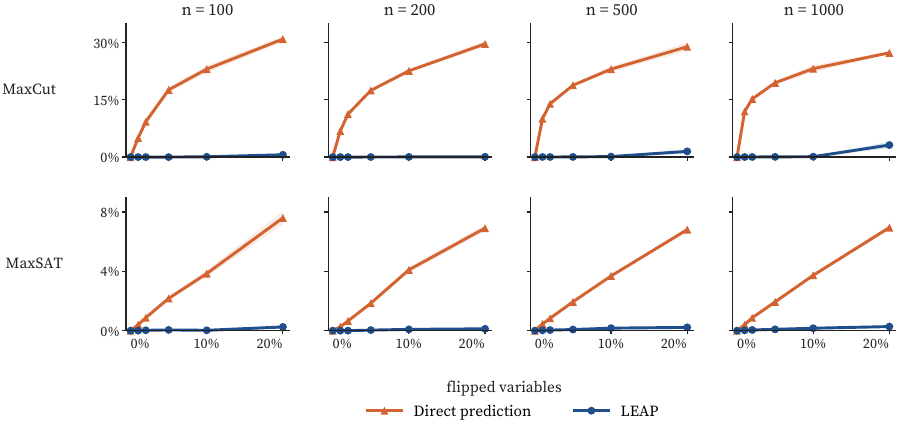}
  \caption{
  Relative objective gaps versus the fraction of flipped variables for \emph{sparse} instances.
  \LEAP{} substantially reduces the gap of the direct prediction at every positive flip fraction and instance size; its largest residual gaps occur on the largest \MAXCUT{} instances.}
  \label{fig:large-flip-repair}
\end{figure}

\FloatBarrier
\subsection{Empirical Selection over Oracle Candidates}\label{sec:empirical-selection-oracle-candidates}

We next study an empirical version of the learning question in \hyperref[q:learned-oracle-performance]{Q2}.
We consider four finite oracle classes: annealing SLS, WalkSAT-style SLS, ProbSAT-style SLS, and GNN predictors.
Within each class, ERM selects the candidate whose induced algorithm has the smallest empirical cost on training instances, and we evaluate the selected candidate on test instances.

We use synthetic \MAXCUT{} and unweighted \MAXSAT{} instances with \(n=200\) variables and vary the number \(m\) of training instances from \(1\) to \(3000\) (see \cref{app:experimental-details} for details).
For each oracle class \(\mathcal F\), \cref{fig:erm-method-family-selection} plots the excess test risk of the ERM-selected candidate \(\hat f_m\), that is, its average test cost minus the smallest average test cost over \(\mathcal F\).

\begin{figure}[!htbp]
  \centering
  \includegraphics[width=\linewidth]{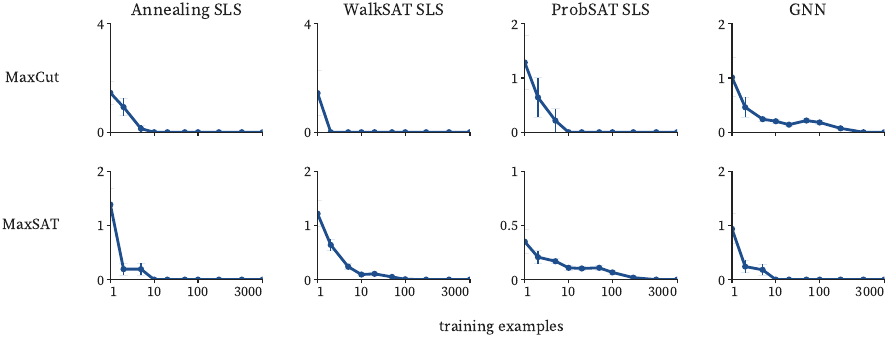}
  \caption{
  ERM selects useful oracle candidates from training instances.
  The horizontal axis shows the number of training instances; the vertical axis shows excess test risk, the selected candidate's average test cost minus the smallest average test cost in its oracle class.
  Error bars show standard errors over repeated draws of the training instances.}
  \label{fig:erm-method-family-selection}
\end{figure}

\Cref{fig:erm-method-family-selection} shows that the excess test risk decreases steadily as ERM is given more training instances, consistent with \cref{thm:erm-learnability}.

\subsection{Large \MAXSAT{} Benchmarks}\label{sec:sota-maxsat}

\Cref{sec:solution-quality,sec:empirical-selection-oracle-candidates} evaluate our two theoretical guarantees on controlled instances.
Finally, we stress-test the relaxation-and-rounding pipeline of \LEAP{} at scale on \MAXSAT{} application benchmarks from eight families, comparing its solution quality with nine state-of-the-art \MAXSAT{} solver configurations~\citep{MSE2024Solvers} on the \(256\) instances for which all methods return an assignment.
Benchmark suites of this kind consist of heterogeneous instances and provide no training distribution from which to learn an oracle.
We therefore run \LEAP{} in an \emph{iterative self-refinement} mode: starting from a uniformly random assignment, \LEAP{} runs ten rounds, each using the previous output as its prediction and the doubling grid of \cref{prop:doubling-grid}.
Each round solves \(O(\log n)\) linear programs, so \LEAP{} takes longer than the \(60\)-second baselines; the comparison therefore focuses on solution quality.
\Cref{tab:maxsat-sizes} summarizes the sizes of these instances.
We report the best of five random-start runs for \LEAP{} and one \(60\)-second run per baseline; a single \LEAP{} run takes under an hour on most instances, with the longest runs taking tens of hours (see \cref{app:sota-maxsat-details} for details).
The clause-normalized gap is the number of clauses left unsatisfied by \LEAP{} minus that of the best observed baseline, as a percentage of the \(M\) clauses.

\begin{table}[!htbp]
  \centering
  \small
  \begin{tabular}{@{}lrrrr@{}}
    \toprule
    \textbf{Family} & \textbf{Instances} & \textbf{Variables} & \textbf{Clauses} & \textbf{Width} \\
    \midrule
    Deezer & 40 & 25 / 156 & 169 / 1\,684 & 2.00 / 2 \\
    Decision trees & 14 & 27\,015 / 82\,071 & 231\,206 / 660\,735 & 5.27 / 98 \\
    BDD optimization & 23 & 6\,491 / 21\,190 & 453\,237 / 1\,887\,499 & 2.75 / 128 \\
    Rule-based classifiers & 38 & 27\,611 / 141\,265 & 359\,859 / 1\,970\,958 & 2.81 / 97 \\
    Single-machine scheduling & 31 & 26\,343 / 43\,672 & 734\,582 / 1\,775\,621 & 2.97 / 3 \\
    Judgment aggregation & 38 & 12\,894 / 22\,994 & 296\,276 / 1\,567\,566 & 2.89 / 3 \\
    Access control & 26 & 1\,877 / 2\,050 & 5\,043 / 6\,900 & 2.44 / 23 \\
    Knowledge-base inconsistency & 46 & 64\,919 / 326\,096 & 156\,135 / 648\,235 & 2.61 / 121 \\
    \bottomrule
  \end{tabular}
  \normalsize
  \caption{Instance sizes of the eight benchmark families. Entries are mean / maximum over instances; width counts the literals in a clause.}
  \label{tab:maxsat-sizes}
\end{table}

\begin{figure}[!htbp]
  \centering
  \includegraphics[width=0.87\linewidth]{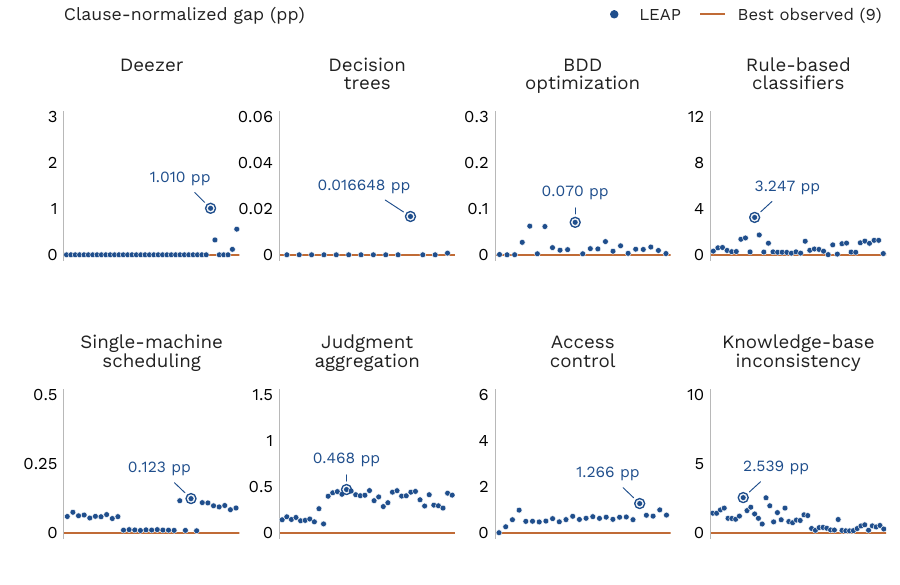}
  \caption{Clause-normalized gaps (percentage points) to the best of nine state-of-the-art solver configurations on each instance. Instances are ordered by size; annotations mark maxima.}
  \label{fig:sota-maxsat}
\end{figure}

\Cref{fig:sota-maxsat} shows that the mean clause-normalized gap is below one percentage point in every family.
Together with \cref{tab:maxsat-sizes}, this shows that \LEAP{} scales to instances with up to about two million clauses and, even when started from random assignments, stays close to the best observed baseline.

\FloatBarrier
\section{Concluding Remarks}\label{sec:conclusion}
This work reveals that a learnable prediction model can guide hard smooth integer programs with provable performance guarantees.
The loss of \LEAP{} is governed by the error footprint of the prediction, that is, by how strongly its errors perturb the local coefficients of the decomposed objective, rather than by the number of errors alone.
This sensitivity depends on the structure of the instance and on the decomposition chosen by the algorithm: a suitable variable order limits how far errors propagate, errors that perturb no local coefficient do not enter the guarantee at all, and unless \(\P=\NP\) no polynomial-time algorithm improves the order of the resulting local bound.
The experiments are consistent with this picture: \LEAP{} substantially repairs imperfect predictions on controlled instances and scales to large \MAXSAT{} application benchmarks.

This view of prediction error also bears on how oracles should be learned.
Since errors at different coordinates have different downstream costs, an oracle is naturally evaluated by the algorithm it induces rather than by coordinate-wise accuracy alone, and our pseudo-dimension bound gives a sample-complexity guarantee for selecting it by empirical induced cost.
The main open question lies on the computational side: how to carry out this selection, or the training it calls for, efficiently for expressive predictor classes.
\clearpage
\appendix
\section*{\centering\Large Appendix}
This appendix is organized as follows:
\begin{itemize}
    \item \Cref{sec:related-works-appendix} gives extended related-work discussion.
    \item \Cref{app:experimental-details} gives additional experimental details.
    \item \Cref{app:omitted-proofs} gives omitted approximation-framework proofs, in the order of the main text: footprint control and the doubling grid, the positivity of \(\kappa_b^*\) and its slack certificates, the footprint bounds, the structural gain, and the worst-case lower bound.
    \item \Cref{sec:omitted-proofs-learnability} gives omitted proofs for the learnability guarantees.
    \item \Cref{sec:maxcsp-smoothness} verifies the smoothness condition for \MAXKCSP{}.
\end{itemize}

\section{Extended Related Work}\label{sec:related-works-appendix}

We compare approximation results through their structural assumptions, advice models, and guarantees.

\subsection{Detailed Comparison with Approximation Results}\label{sec:detailed-approx-comparison}

Throughout this subsection, \(q=\Pr[\text{advice bit correct}]-1/2\), \(\varepsilon=\|\widehat{\x}-\x^*\|_1\), and \(r\) denotes the fraction of variables revealed by subset advice.

\paragraph{Classical Approximation and Density.}
For general \MAXCUT{} without advice, \citet{goemans1995improved} achieve approximation ratio \(\alpha_{\mathrm{GW}}\approx0.87856\) by solving an SDP relaxation and applying random-hyperplane rounding.
Density assumptions permit stronger guarantees.
For dense \MAXCUT{} and \MAXKSAT{}, where the optimum is \(\Omega(n^2)\) and \(\Omega(n^k)\), respectively, \citet{Arora1999PTAS} obtain a \((1-\eta)\)-approximation for any fixed \(\eta>0\) through exhaustive sampling, LP relaxation, and rounding.
The running time is commonly written as \(n^{O(1/\eta^2)}\), with density parameters suppressed in the exponent.
\citet{Fotakis2016CSP} extend this approach to almost-sparse \MAXKCSP{} instances with \(\OPT=\Omega(n^{k-1+\gamma})\) for fixed \(\gamma\in(0,1]\), obtaining a \((1-\eta)\)-approximation in \(2^{O(n^{1-\gamma}\log n/\eta^3)}\) time.

\paragraph{Sampled Advice.}
\citet{bampis24PLA} incorporate predictions into the sampling framework through a parsimonious one-bit oracle, queried on a sampled set \(S\).
For \(\rho\)-dense \MAXCUT{}, the prediction-dependent ratio is \(1-\eta-8\,\mathrm{err}/(\rho|S|)\), where \(\mathrm{err}\) counts incorrect predictions in \(S\).
Their running time is \(O(nT_{\mathrm{LP}})\), where \(T_{\mathrm{LP}}\) is the time to solve an LP with \(n\) variables and \(O(n)\) constraints.
For general smooth integer programs of degree \(d\) and smoothness parameter \(\beta\), their method requires \(O(n\log(\beta nd))\) LP solves.

\paragraph{Label and Edge Advice.}
\citet{Cohen2024MaxCut} obtain an \(\alpha_{\mathrm{GW}}+\widetilde{\Omega}(q^4)\) approximation for \MAXCUT{} from label advice with bias \(q\), with pairwise independence sufficient for their guarantee.
Under independent label advice and sufficiently large average degree \(\bar d=2|E|/n\), \citet{ghoshal2025constraint} obtain expected value at least \((1-O(1/(q\sqrt{\bar d})))\OPT\) for unweighted \MAXCUT{} and Max-2-Lin.
With independent \(q\)-correlated edge advice, \citet{Aamand2025Improved} achieve \(\alpha_{\mathrm{GW}}+\widetilde{\Omega}(q^2)\) for \MAXCUT{} by aggregating predictions at high-degree vertices and using an approximation subroutine for the remaining part.

\paragraph{Subset Advice.}
Subset advice reveals the optimal labels of a random fraction \(r\) of variables.
For \MAXCUT{}, \citet{Cohen2024MaxCut} obtain expected ratio \(\alpha_{\mathrm{RT}}+\Omega(r)\), where \(\alpha_{\mathrm{RT}}\approx0.858\) is the max-bisection ratio of \citet{raghavendra2012approximating}.
For \MAXSAT{}, \citet{Attias2025LAA-SAT} show that simplifying the formula with these labels and applying an \(\alpha\)-approximation to the remaining instance yields ratio \(\alpha+(1-\alpha)r\).

\paragraph{Relationship to Footprint Analysis.}
Our framework uses a full assignment prediction and controls its effect through the error footprint \(\Phi_\varepsilon(\widehat{\x})\), without requiring independence of prediction errors.
For \MAXCUT{} (\(k=2\)) and \MAXKSAT{}, global smoothness gives ratio \(1-O(\varepsilon/n^\gamma)\) when \(\OPT=\Omega(n^{k-1+\gamma})\), and the local bound gives \(1-O(\Delta\varepsilon/n)\) when each vertex has degree at most \(\Delta\) or each literal occurs in at most \(\Delta\) clauses, and more generally \(1-O(\dgn\cdot\varepsilon/n)\) in terms of the degeneracy of the instance.
With the doubling grid of \cref{prop:doubling-grid}, \LEAP{} requires \(O(\log n)\) LP solves and retains these approximation orders.

\subsection{Broader Learning-Augmented Algorithms}

Learning-augmented algorithms use predictions to improve performance while studying consistency, degradation with prediction error, and robustness~\citep{Lykouris2021LearningAugmented, mitzenmacher2022algorithms}.
Applications include scheduling~\citep{Kumar2018ImprovingOnlineAlgorithms, Mitzenmacher2020SchedulingWithPredictions, shahout2024scheduling}, matching~\citep{Dinitz2021Matching}, sorting~\citep{Bai2023Sorting}, clustering~\citep{ergun2022learningaugmented}, indexing~\citep{Kraska2018LearnedIndex}, shortest paths~\citep{Lattanzi2023Speeding}, paging~\citep{Antoniadis2023LAA-Paging}, caching~\citep{Lykouris2018LearningAugmented, Im2022Parsimonious, chen2025caching}, and branch-and-bound~\citep{Balcan2018LearningToBranch}.
\section{Additional Experimental Details}\label{app:experimental-details}

This appendix provides construction and protocol details for the experiments in \cref{sec:experiments}.

\subsection{Controlled Perturbation Experiments}\label{app:controlled-perturbations}

For each prediction in either perturbation experiment, we enumerate all integer budgets \(b\in\{0,\ldots,n\}\).
We build the relaxation from the \(\sigma\)-decomposition with \(\sigma\) the natural variable order.
For each budget, we solve the footprint relaxation of \cref{optimization:d-LP-oracle} and apply deterministic rounding in natural variable order, rounding ties to 1.
We return the rounded assignment with the largest original objective value, breaking ties in favor of the smaller budget.
For both constructions, the plotted flip fractions are $0\%,1\%,2\%,5\%,10\%,20\%$.

\paragraph{Dense constructions.}
For \MAXCUT{}, we draw a balanced planted bipartition $\x^*$ and sample $\lceil 0.1n^2\rceil$ crossing edges, so that $\OPT=|E|$.
For \MAXSAT{}, we draw a planted assignment $\x^*$ and sample $\lceil 0.002n^3\rceil$ distinct 3-clauses satisfied by $\x^*$, so that $\OPT=|\mathcal C|$.
We use five independent seeds per problem and size.
Within each instance, we fix one random variable ordering and flip progressively longer prefixes of this ordering to create nested perturbations.

\paragraph{Sparse constructions.}
For \MAXCUT{}, we sample \(2n\) edges across a planted random bipartition, using nonuniform vertex sampling to obtain heterogeneous degrees.
For \MAXSAT{}, we generate \(2n\) distinct planted-satisfied 3-clauses, with each variable occurring at most eight times.
Thus both constructions have \(\OPT=2n\).
For each problem and size in \cref{fig:large-flip-repair}, we generate five independent planted instances.
Within each instance, we average four perturbations at each flip fraction: two random flip sets, a highest-load set, and a lowest-load set.
Load is graph degree for \MAXCUT{} and variable occurrence count for \MAXSAT{}.
Curves average the five instance-level means.

\subsection{Empirical Selection over Oracle Candidates}\label{app:empirical-selection-details}

For each train-test setting, we use $3000$ training instances and $2000$ test instances.
\MAXCUT{} instances are Barab\'asi--Albert graphs, and \MAXSAT{} instances are random 3-CNF formulas with a fixed clause-to-variable ratio.
Cost counts uncut edges for \MAXCUT{} and unsatisfied clauses for \MAXSAT{} in the assignment returned by \LEAP{}.
For \MAXCUT{}, CNF-based candidates use an XOR-CNF encoding, in which each edge \((u,v)\) is represented by clauses enforcing \(x_u\ne x_v\); all assignments are scored under the original graph objective.

\Cref{tab:oracle-candidates} summarizes the four candidate classes.
Each local-search variant uses time budgets of \(0.0025,0.005,0.01,0.025\) seconds.
For each neural architecture, checkpoints after \(1,2,4\) training epochs define the candidates.
The neural training instances are random 3-CNF formulas for \MAXSAT{} and XOR-CNF encodings of Barab\'asi--Albert graphs for \MAXCUT{}.

\begin{table}[!htbp]
  \centering
  \small
  \setlength{\tabcolsep}{4pt}
  \begin{tabularx}{\linewidth}{@{}lXr@{}}
    \toprule
    \textbf{Class} & \textbf{Variants} & \textbf{Count} \\
    \midrule
    Annealing SLS & Simulated annealing~\citep{kirkpatrick1983annealing} & 4 \\
    WalkSAT-style SLS & GSAT, MaxWalkSAT, weighted and restart MaxWalkSAT~\citep{selman1992gsat,selman1996walksat} & 16 \\
    ProbSAT-style SLS & ProbSAT-like and SATLike-style search~\citep{balint2012probsat} & 8 \\
    GNN predictors & GGNN, NeuroSAT, and GCN~\citep{li2016ggnn,Selsam2019LBLMD,kipf2017gcn} & 9 \\
    \bottomrule
  \end{tabularx}
  \normalsize
  \caption{Oracle candidates within each class.}
  \label{tab:oracle-candidates}
\end{table}

For each problem, oracle class, and training size $m$, we repeat the training draw $50$ times and report means and standard errors.

\paragraph{Selection at a Fixed Instance Size.}
\Cref{fig:erm-selection-extra-n100,fig:erm-selection-extra-n500,fig:erm-selection-extra-n1000} show results for $n=100, 500, 1000$.
These plots generally show the selected candidates' excess risk declining toward zero as sample counts increase, though some families exhibit intermediate fluctuations.

\begin{figure}[!tbp]
  \centering
  \includegraphics[width=0.92\linewidth]{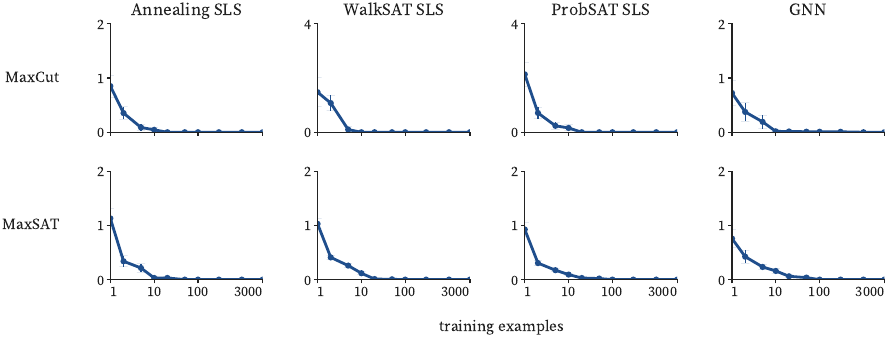}
  \caption{Additional empirical-selection results for \(n=100\). The horizontal axis shows the number of training instances; the vertical axis shows excess test risk.}
  \label{fig:erm-selection-extra-n100}
\end{figure}

\begin{figure}[!tbp]
  \centering
  \includegraphics[width=0.92\linewidth]{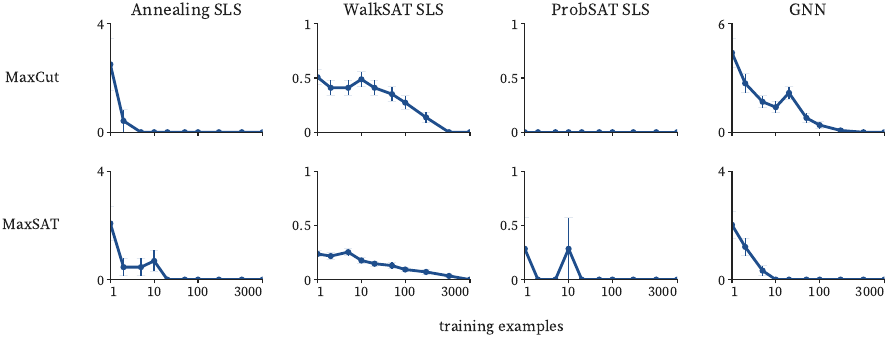}
  \caption{Additional empirical-selection results for \(n=500\). The horizontal axis shows the number of training instances; the vertical axis shows excess test risk.}
  \label{fig:erm-selection-extra-n500}
\end{figure}

\begin{figure}[!tbp]
  \centering
  \includegraphics[width=0.92\linewidth]{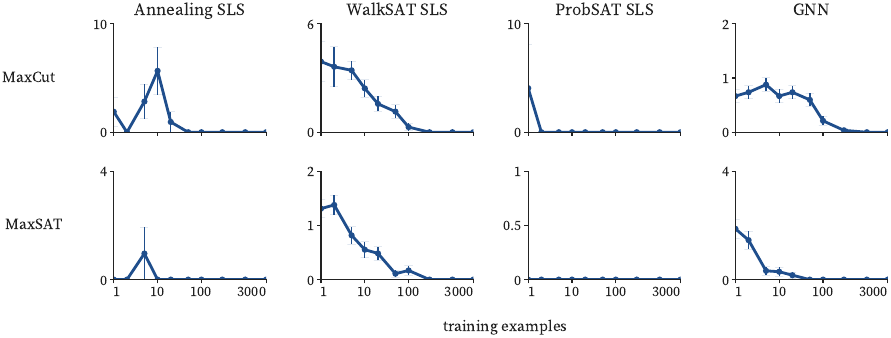}
  \caption{Additional empirical-selection results for \(n=1000\). The horizontal axis shows the number of training instances; the vertical axis shows excess test risk.}
  \label{fig:erm-selection-extra-n1000}
\end{figure}

\FloatBarrier
\paragraph{Selection across Instance Sizes.}
We evaluate candidates selected on smaller training instances and tested on larger instances.
\Cref{fig:erm-selection-extra-n100-n200,fig:erm-selection-extra-n500-n1000} show results for extrapolation.
The selected candidates approach zero excess risk, though a persistent positive gap is observed for the ProbSAT family when transferring from $n=100$ to $n=200$ in \MAXSAT{}.

\begin{figure}[!htbp]
  \centering
  \includegraphics[width=0.92\linewidth]{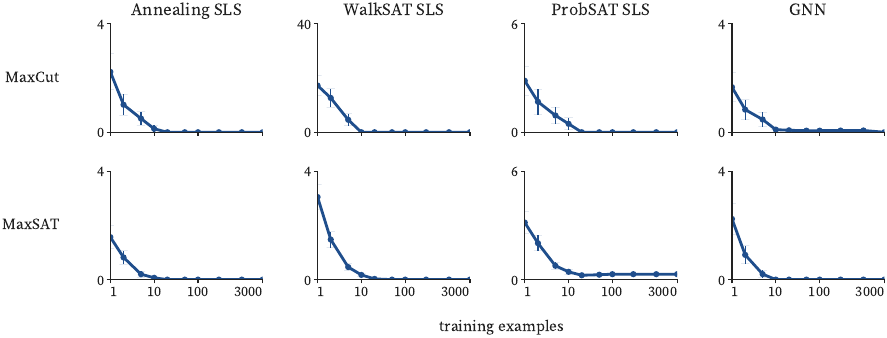}
  \caption{Empirical-selection results for extrapolation from \(n=100\) to \(n=200\). The horizontal axis shows the number of training instances; the vertical axis shows excess test risk on the larger test instances.}
  \label{fig:erm-selection-extra-n100-n200}
\end{figure}

\begin{figure}[!htbp]
  \centering
  \includegraphics[width=0.92\linewidth]{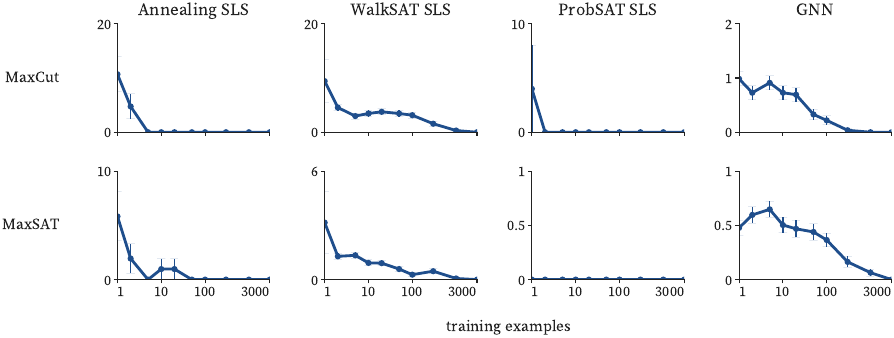}
  \caption{Empirical-selection results for extrapolation from \(n=500\) to \(n=1000\). The horizontal axis shows the number of training instances; the vertical axis shows excess test risk on the larger test instances.}
  \label{fig:erm-selection-extra-n500-n1000}
\end{figure}

\subsection{Details of the Large-Scale \MAXSAT{} Benchmarks}\label{app:sota-maxsat-details}

Seven families are drawn from the historical anytime-unweighted or incomplete-unweighted tracks of the \MAXSAT{} Evaluation archive~\citep{Jabs2026BenchmarkCollection}; the eighth consists of \(40\) Deezer ego networks from SNAP~\citep{Rozemberczki2020KarateClub}.
\Cref{tab:maxsat-apps} describes their source applications.
After deduplication by normalized formula content, we retain \(261\) instances.
The common comparison contains \(256\) instances; five BDD instances lack assignments from at least one baseline.

\begin{table}[!htbp]
  \centering
  \small
  \setlength{\tabcolsep}{4pt}
  \begin{tabularx}{\linewidth}{@{}>{\raggedright\arraybackslash}p{1.55in}X@{}}
    \toprule
    \textbf{Family} & \textbf{Source application} \\
    \midrule
    Deezer & Partition social ego networks to maximize crossing edges. \\
    Decision trees & Learn accurate decision trees under depth and size bounds. \\
    BDD optimization & Learn bounded-depth binary decision diagram classifiers. \\
    Rule-based classifiers & Learn Boolean rules balancing classification accuracy and complexity. \\
    Single-machine scheduling & Maximize on-time jobs subject to release times, processing requirements, and dependencies. \\
    Judgment aggregation & Find consistent collective rankings minimizing disagreement with voter or majority preferences. \\
    Access control & Select roles granting required permissions under permission and incompatibility limits. \\
    Knowledge-base inconsistency & Quantify inconsistency through violated statements, contradictory values, or distances to satisfying assignments. \\
    \bottomrule
  \end{tabularx}
  \normalsize
  \caption{Source applications of the eight benchmark families. Their encodings are evaluated as unit-weight all-soft formulas.}
  \label{tab:maxsat-apps}
\end{table}

All methods use the same formula.
For Deezer, the XOR-CNF encoding of each edge makes the cost the number of uncut edges and gives \(M=2|E|\).
Family means average the per-instance clause-normalized gaps.

We use the MSE 2024 versions of Exact, Loandra, NuWLS-c-IBR, SPB-MaxSAT-c, SPB-MaxSAT-c-Band, SPB-MaxSAT-c-FPS, noSAT-MaxSATv3, and the Glucose and IntelSAT variants of tt-open-wbo-inc~\citep{MSE2024Solvers}.
Baselines receive \(60\) seconds and \(12\)~GiB per instance.
\LEAP{} performs ten rounds from each of five random starts, using the output of each round as the prediction for the next, and returns the best final assignment by cost.
Each round uses the doubling grid of \cref{prop:doubling-grid}.
\section{Omitted Material for the Approximation Framework}
\label{app:omitted-proofs}

This section collects the proofs omitted from \cref{sec:approach,sec:applications}, in the order in which the results appear.

\subsection{Proofs for the Error-Footprint Guarantee}\label{app:footprint-guarantee}
\begin{proof}[Proof of \cref{lemma:footprint-control}.]
For a node \(I\), let \(\mathcal T(I)\) be the set of non-root constrained descendants of \(I\) in the decomposition tree.
We prove, by induction on the remaining degree of \(p_I\), that every feasible \((\y,\mathbf u)\) satisfies
\[
  |p_I(\y)-q_I(\y)|
  \le
  \sum_{J\in\mathcal T(I)}\tau_J(\mathbf u).
\]
If \(p_I\) is linear, then every nonzero child component is constant, so \(p_I(\y)=q_I(\y)\).

For the induction step,
\[
  p_I(\y)-q_I(\y)
  =
  \sum_{j\in[n]}y_j\bigl(p_{I,j}(\y)-p_{I,j}(\widehat{\x})\bigr).
\]
Constant child components contribute zero to this difference.
For each nonconstant child, we add and subtract its surrogate \(q_{I,j}(\y)\).
The feasibility constraint gives
\[
  |q_{I,j}(\y)-p_{I,j}(\widehat{\x})|\le \tau_{I,j}(\mathbf u),
\]
and the induction hypothesis bounds \(|p_{I,j}(\y)-q_{I,j}(\y)|\) by the tolerances in the child subtree.
Summing over children gives exactly the tolerances over \(\mathcal T(I)\).

Applying this recurrence at the root, where \(\mathcal T(\emptyset)=\mathcal I\), to \(\y\in P_b\) with a certificate \(\mathbf u\) satisfying \(\sum_ju_j\le b\) yields
\[
  |\zeta(\y)|
  \le
  \sum_{I\in\mathcal I}\tau_I(\mathbf u)
  =
  \sum_{j\in[n]}F_j(\widehat{\x})u_j
  \le
  \Phi_b(\widehat{\x}).
\]

For the second claim, let \(S=\{j:\widehat{x}_j\ne x^*_j\}\), so that \(|S|=\varepsilon\), and set \(u^*_j=\mathbf 1[j\in S]\).
Since \(q_I(\widehat{\x})=p_I(\widehat{\x})\), every \(I\in\mathcal I\) satisfies
\[
  |q_I(\x^*)-p_I(\widehat{\x})|
  =
  |q_I(\x^*)-q_I(\widehat{\x})|
  \le
  \sum_{j\in S}|p_{I,j}(\widehat{\x})|
  =
  \tau_I(\mathbf u^*).
\]
Thus \((\x^*,\mathbf u^*)\) is feasible with budget \(\varepsilon\), so \(\x^*\in P_\varepsilon\), and the first claim gives \(|\zeta(\x^*)|\le\Phi_\varepsilon(\widehat{\x})\).
Increasing the budget only relaxes the constraint \(\sum_ju_j\le b\), so \(P_\varepsilon\subseteq P_b\) for every \(b\ge\varepsilon\).
\end{proof}

\begin{proof}[Proof of \cref{prop:doubling-grid}.]
If \(\varepsilon=0\), the budget \(0\in\mathcal B\) gives \(p(\z)\ge\OPT\) because \(\Phi_0(\widehat{\x})=0\).
Otherwise, \(\mathcal B\) contains a budget \(b\) with \(\varepsilon\le b<2\varepsilon\), so \(\x^*\in P_b\) by \cref{lemma:footprint-control}.
Scaling a maximizer of \(\Phi_b(\widehat{\x})\) by \(\varepsilon/b\) shows
\[
  \Phi_\varepsilon(\widehat{\x})\le\Phi_b(\widehat{\x})\le2\Phi_\varepsilon(\widehat{\x}),
\]
and the first bound of \cref{thm:error-footprint-guarantee} gives
\[
  \OPT-p(\z_b)
  \le\Phi_\varepsilon(\widehat{\x})+\Phi_b(\widehat{\x})
  \le3\Phi_\varepsilon(\widehat{\x}).
\]
The algorithm returns a candidate at least as good as \(\z_b\).
\end{proof}

\subsection{Positivity of \texorpdfstring{\(\kappa_b^*\)}{kappa*} and Slack Certificates}\label{app:positive-contraction}
Whether \(\kappa_b^*\) is positive is decided on the optimal face \(\mathcal O_b=\argmax_{\x\in P_b}q(\x)\) of the projected LP.

\begin{proposition}[Positivity of \(\kappa_b^*\)]\label{prop:positive-contraction}
For every budget \(b\),
\[
  \kappa_b^*>0
  \quad\Longleftrightarrow\quad
  \y_b\in\argmax_{\mathbf s\in \mathcal O_b}p(\mathbf s).
\]
In particular, \(\kappa_b^*>0\) whenever \(\mathcal O_b\) consists of the single point \(\y_b\).
\end{proposition}

\begin{proof}
Since \(q=Q_b\) on \(\mathcal O_b\), the right-hand condition is equivalent to \(\zeta(\mathbf s)\le \zeta(\y_b)\) for every \(\mathbf s\in \mathcal O_b\).
If \(\kappa_b^*>0\), then \cref{eqn:contraction-coefficient} at \(\x=\mathbf s\in \mathcal O_b\), where \(g_b(\mathbf s)=0\), forces \(\zeta(\mathbf s)\le \zeta(\y_b)\).

Conversely, assume \(\zeta\le \zeta(\y_b)\) on \(\mathcal O_b\).
If \(\mathcal O_b=P_b\), then \(\zeta(\x)-\zeta(\y_b)\le0\) on \(P_b\), and \(\kappa_b^*=1\).
Otherwise, fix a norm, let \(D>0\) be the diameter of the polytope \(P_b\), and let \(\delta>0\) be the smallest gap \(Q_b-q(\mathbf v)\) over the vertices \(\mathbf v\) of \(P_b\) outside \(\mathcal O_b\).
Write \(\x\in P_b\) as a convex combination of vertices, and let \(t\) be its weight on vertices outside \(\mathcal O_b\).
Moving this weight to a vertex of \(\mathcal O_b\) gives \(\operatorname{dist}(\x,\mathcal O_b)\le tD\), while \(g_b(\x)\ge t\delta\) because \(q\) is affine; hence \(g_b(\x)\ge\mu\operatorname{dist}(\x,\mathcal O_b)\) with \(\mu=\delta/D\).
Since \(\zeta\) is a polynomial, it is \(L\)-Lipschitz on \([0,1]^n\) for some \(L\ge0\).
For \(\x\in P_b\), let \(\mathbf s\in \mathcal O_b\) be a nearest point; then
\[
  \zeta(\x)-\zeta(\y_b)
  \le
  \zeta(\x)-\zeta(\mathbf s)
  \le
  L\operatorname{dist}(\x,\mathcal O_b)
  \le
  \frac{L}{\mu}\,g_b(\x).
\]
Hence \(\kappa_b^*\ge\min\{1,\mu/L\}>0\), where \(\mu/L=+\infty\) if \(L=0\).
\end{proof}

Lower bounds on \(\kappa_b^*\) can be certified from any dual vector that represents the surrogate gain as in \cref{lemma:lp-gain}.

\begin{lemma}[Slack certificate]\label{lemma:slack-certificate}
Let \(\boldsymbol\lambda\ge0\) satisfy \(g_b(\x)=\boldsymbol\lambda^\top\boldsymbol\sigma(\mathbf w)\) for every feasible \(\mathbf w=(\x,\mathbf u)\) at budget \(b\); by \cref{lemma:lp-gain}, every optimal dual solution does.
Suppose that a vector \(\mathbf t\ge0\) satisfies
\[
  \zeta(\x)-\zeta(\y_b)\le\mathbf t^\top\boldsymbol\sigma(\mathbf w)
  \qquad\text{for every feasible }\mathbf w=(\x,\mathbf u).
\]
Then
\[
  \kappa_b^*
  \ge
  \min\Bigl\{1,\ \min_{r:\,t_r>0}\frac{\lambda_r}{t_r}\Bigr\},
\]
where the inner minimum is \(+\infty\) if \(\mathbf t=0\).
\end{lemma}

\begin{proof}
Let \(\kappa\) denote the right-hand side.
Then \(\boldsymbol\lambda\ge\kappa\mathbf t\) componentwise, and since slacks are nonnegative, \(g_b(\x)=\boldsymbol\lambda^\top\boldsymbol\sigma(\mathbf w)\ge\kappa\,\mathbf t^\top\boldsymbol\sigma(\mathbf w)\ge\kappa\bigl(\zeta(\x)-\zeta(\y_b)\bigr)\).
Every \(\x\in P_b\) has a feasible lift \(\mathbf w\), so \(\kappa\) is admissible in \cref{eqn:contraction-coefficient}.
\end{proof}

Rescaling a constraint row by a positive factor divides both its dual price and its weight in \(\mathbf t\) by that factor, so the certified value does not depend on how the constraints are scaled.

\subsection{Proofs of the Footprint Bounds}\label{app:global-footprint}
The proof of \cref{lemma:global-footprint-bound} uses the following magnitude bound for smooth polynomials.

\begin{lemma}[Lemma 3.2, \citealt{Arora1999PTAS}]\label{lemma:polynomial-bound-arora}
  Let \(p(\x)\) be an \(n\)-variate \(\beta\)-smooth polynomial of degree at most \(d\). If \(n>d\), then for any \(\x\in[0,1]^n\), it holds that \(|p(\x)|\le 2\,\beta e\, n^d\).
\end{lemma}

\begin{proof}[Proof of \cref{lemma:global-footprint-bound}.]
Let \(\mathcal I_r\subseteq\mathcal I\) be the set of decomposition tuples of length \(r\).
The root objective is not a component constraint, so the relevant levels are \(r=1,\dots,d-1\).
Fix any \(\mathbf u\in[0,1]^n\) satisfying \(\sum_j u_j\le\varepsilon\).

For \(1\le r<d-1\), each child component \(p_{I,j}\) with \(I\in\mathcal I_r\) lies at depth \(r+1\), so by \cref{lemma:decomposition-of-polynomials} it has degree at most \(d-r-1\) and is \(\beta\)-smooth with respect to this degree bound.
By \cref{lemma:polynomial-bound-arora}, applied with degree bound \(d-r-1\),
\[
  |p_{I,j}(\widehat{\x})|\le 2e\beta n^{d-r-1}.
\]
There are at most \(n^r\) tuples in \(\mathcal I_r\), so the level-\(r\) contribution is at most
\[
  \sum_{I\in\mathcal I_r}\sum_{j\in[n]}|p_{I,j}(\widehat{\x})|u_j
  \le
  2e\beta n^{d-1}\varepsilon .
\]
There are \(d-2\) such levels.
At the final level \(r=d-1\), the child components lie at depth \(d\) and are constant coefficients with magnitude at most \(\beta\), and therefore
\[
  \sum_{I\in\mathcal I_{d-1}}\sum_{j\in[n]}|p_{I,j}(\widehat{\x})|u_j
  \le
  \beta n^{d-1}\varepsilon .
\]
Summing over all levels gives
\[
  \sum_{j\in[n]}F_j(\widehat{\x})u_j
  \le
  \bigl(2e(d-2)+1\bigr)\beta n^{d-1}\varepsilon.
\]
Taking the maximum over feasible \(\mathbf u\) proves the claim for any \(\rho_d\ge 2e(d-2)+1\).
\end{proof}

The ordered-load bound follows by tracking where the \(\sigma\)-decomposition charges each monomial.

\begin{proof}[Proof of \cref{lemma:coordinate-load-footprint-bound}.]
Write a nonconstant monomial as \(a_Mx_{i_1}\cdots x_{i_s}\) with \(i_1,\dots,i_s\) in the order \(\sigma\).
The \(\sigma\)-decomposition factors out \(i_1\) at the root, which is not constrained.
If \(s\ge2\), it then factors out each later variable \(i_{t+1}\), \(1\le t\le s-1\), at the node \(I=(i_1,\dots,i_t)\), where the monomial contributes \(a_Mx_{i_{t+2}}\cdots x_{i_s}\) to the child \(p_{I,i_{t+1}}\).
This node is constrained, since distinct monomials leave distinct quotients and \(p_I\) therefore contains the nonconstant term \(a_Mx_{i_{t+1}}\cdots x_{i_s}\).
Since \(\widehat{\x}\) is Boolean, the contribution to \(p_{I,i_{t+1}}(\widehat{\x})\) has magnitude at most \(|a_M|\).
Each monomial \(M\) therefore adds at most \(|a_M|\) to \(F_j(\widehat{\x})\) for every \(j\in M\setminus\{\min_\sigma M\}\), and the triangle inequality gives \(F_j(\widehat{\x})\le L_j^\sigma\).

Finally, the maximization defining \(\Phi_b(\widehat{\x})\) is a fractional knapsack problem with unit costs and integral budget \(b\), so its value is the sum of the \(b\) largest footprint weights, which is at most \(\sum_{t=1}^{b}L^\sigma_{(t)}\).
\end{proof}

\subsection{Proof of the Structural Gain}\label{app:structural-gain}
\begin{proof}[Proof of \cref{prop:structural-gain}.]
We view the prediction error through the coordinates that actually move a local coefficient.
Let \(S^+=\{j\in S:\ F_j(\widehat{\x})>0\}\), and call \(\varepsilon^+=|S^+|\le\varepsilon\) the \emph{effective error}.
For \(j\in S\setminus S^+\), \(F_j(\widehat{\x})=0\) forces \(p_{I,j}(\widehat{\x})=0\) for every \(I\in\mathcal I\), so these coordinates move no local coefficient.
The argument for the second claim of \cref{lemma:footprint-control}, with \(S\) replaced by \(S^+\), therefore shows that \((\x^*,\mathbf 1_{S^+})\) is feasible at budget \(\varepsilon^+\), and the first claim gives \(\zeta(\x^*)\le\sum_{j\in S^+}F_j(\widehat{\x})\le\Phi_{\varepsilon^+}(\widehat{\x})\).
The proof of \cref{thm:error-footprint-guarantee} at budget \(\varepsilon^+\) then yields \(\OPT-p(\z)\le2\Phi_{\varepsilon^+}(\widehat{\x})\).

Finally, \cref{lemma:coordinate-load-footprint-bound} gives \(F_j(\widehat{\x})\le L_j^\sigma\), so \(L_j^\sigma=0\) implies \(j\notin S^+\) and hence \(\varepsilon^+\le\varepsilon_\sigma\).
The same lemma at budget \(\varepsilon^+\) gives
\[
  \OPT-p(\z)
  \le2\Phi_{\varepsilon^+}(\widehat{\x})
  \le2\sum_{t=1}^{\varepsilon^+}L^\sigma_{(t)}
  \le2\sum_{t=1}^{\varepsilon_\sigma}L^\sigma_{(t)}.
\]
\end{proof}

\subsection{Worst-Case Order of the Local Bound}\label{app:local-order}
We prove \cref{prop:local-order-lower-bound} by reduction from \MAXCUT{} on graphs of maximum degree three, which is \APX-complete~\citep[Theorem~4.2]{Alimonti2000Cubic}.
Let \(H\) be such a graph with \(N\) nonisolated vertices, so \(\OPT(H)\ge|E(H)|/2\ge N/4\).
Unless \(\P=\NP\), polynomial-time loss \(o(N)\) is impossible: together with exhaustive search for bounded \(N\), it would give a PTAS.

\paragraph{Construction.}
Attach a private leaf \(\ell_v\) to each vertex \(v\) of \(H\), obtaining \(H^\dagger\).
For an integer \(t\ge1\), replace each vertex \(v\) by \(t\) independent copies \(V_v\), and each edge by all \(t^2\) edges between its endpoint groups.
Call the resulting graph \(G_t\), and predict the all-zero assignment.

\begin{lemma}[Properties of \(G_t\)]\label{lemma:gt-properties}
For every \(t\ge1\):
\begin{enumerate}[label=(\alph*), leftmargin=*]
  \item Every optimal cut is constant on each \(V_v\) and separates \(V_v\) from \(V_{\ell_v}\); hence \(\varepsilon=tN\).
  \item A cut with loss \(L\) yields, in polynomial time, a cut of \(H\) with loss at most \(L/t^2\).
  \item \(t\le\dgn(G_t)\le\Delta\le4t\).
\end{enumerate}
\end{lemma}

\begin{proof}
Let \(w_v\) be the fraction of ones in \(V_v\), and let \(p_H\) be the multilinear cut polynomial of \(H\).
Then
\[
  \frac{\operatorname{cut}_{G_t}(\x)}{t^2}
  =p_H(\mathbf w|_H)+\sum_v(w_v+w_{\ell_v}-2w_vw_{\ell_v})
  \le\OPT(H)+N.
\]
Here \(p_H\le\OPT(H)\) on \([0,1]^N\), and each leaf term is at most one, with equality only for complementary Boolean values.
An optimal cut of \(H\) with opposite leaf labels attains equality, giving \(\OPT(G_t)=t^2(\OPT(H)+N)\).
Equality forces every leaf term to be one, proving~(a).
For a cut with loss \(L\), the same identity gives \(\OPT(H)-p_H(\mathbf w|_H)\le L/t^2\); deterministic rounding of \(\mathbf w|_H\) proves~(b).
Finally, each copy of \(v\) has degree \(t\deg_{H^\dagger}(v)\le4t\), and every original edge induces a \(K_{t,t}\) of degeneracy \(t\), proving~(c).
\end{proof}

\paragraph{Proof for \MAXCUT{}.}
Apply \cref{lemma:gt-properties}(b).
For part~(i), set \(t=1\): a guarantee \(\psi(\varepsilon)\varepsilon\) gives loss \(\psi(N)N=o(N)\) on \(H\).
For part~(ii), set \(t=N\): then \(\varepsilon=N^2\) and \(N\le\theta\le4N\) for \(\theta\in\{\dgn(G_N),\Delta\}\), so a guarantee \(\psi(\theta)\theta\varepsilon\) gives loss at most \(\psi(\theta)\theta\varepsilon/N^2\le4\psi(\theta)N=o(N)\) on \(H\).
Both constructions have polynomial size, contradicting the hardness above.

\paragraph{Approximation ratios.}
Add \(N'\) disjoint edges to \(G_t\), all predicted optimally, with \(N'\) polynomial in \(N\) and of larger order than both \(|E(G_t)|\) and \(\theta\varepsilon\).
This preserves \(\varepsilon\), \(\Delta\), and \(\dgn\), while \(\OPT,n=\Theta(N')\) and \(\theta\varepsilon/n\to0\).
A ratio \(1-h(\theta\varepsilon/n)\) with \(h(s)=o(s)\) would give loss \(o(\theta\varepsilon)\); restricting to \(G_t\) cannot increase loss, so the same reductions apply.

\paragraph{Extension to \MAXKSAT{}.}
Encode each edge \(\{u,v\}\) by \((x_u\vee x_v)\) and \((\neg x_u\vee\neg x_v)\), giving objective \(|E|+\operatorname{cut}(\x)\).
For \(k\ge3\), extend each \(2\)-clause \(C\) by all \(2^{k-2}\) sign patterns on \(k-2\) fresh private variables; exactly \(2^{k-2}-1+\mathbf 1[C\text{ is satisfied}]\) clauses per block are satisfied.
Thus losses are preserved, no clause is repeated, and predicting the fresh variables as zero leaves the nearest-optimum error unchanged.
For every \(k\ge2\), the maximum literal occurrence count and the formula degeneracy satisfy
\[
  t\le\Delta\le2^{k-2}\cdot4t,
  \qquad t\le\dgn(\phi)\le2^{k+1}t.
\]
For the degeneracy lower bound, fix an edge \(\{u,v\}\) of \(H^\dagger\).
In any order, the last variable in \(V_u\cup V_v\) has \(t\) earlier neighbors, so at least \(t\) clauses containing it start earlier; the upper bound follows from \cref{lemma:maxsat-footprint}.
The reductions therefore carry over with constants depending only on \(k\).

\section{Omitted Material for Learnability Guarantees}\label{sec:omitted-proofs-learnability}

We recall pseudo-dimension and derive from uniform convergence the ERM guarantee used in \cref{sec:learnability}.

\begin{definition}[Pseudo-dimension, \citealt{anthony1999neural}]\label{def:pseudo-dimension}
    Let \(\mathcal{H}\) be a set of real-valued functions defined on \(\Pi\).
    A finite subset \(S = \{\pi_1, \dots, \pi_m\} \subseteq \Pi\) is \emph{(pseudo-)shattered} by \(\mathcal{H}\) if there exist real-valued witnesses \(r_1, \dots, r_m\) such that for each subset \(T \subseteq S\), there exists a function \(h \in \mathcal{H}\) satisfying \(h(\pi_i) > r_i\) if and only if \(\pi_i \in T\).
    The \emph{pseudo-dimension} of \(\mathcal{H}\), denoted by \(\mathrm{Pdim}(\mathcal{H})\), is the maximum cardinality \(m\) for which some subset \(S \subseteq \Pi\) of size \(m\) is pseudo-shattered by \(\mathcal{H}\).
\end{definition}

\begin{samepage}
\begin{lemma}[Uniform convergence]\label{lemma:uniform-convergence}
Let \(H>0\), and let \(\mathcal{H}\) be a nonempty class of functions mapping \(\Pi\) to \([0,H]\), with finite pseudo-dimension \(d_{\mathcal{H}}\). Assume the suprema below are measurable. For any distribution \(\mathcal{D}\) over \(\Pi\), any precision \(\alpha > 0\), and any confidence parameter \(\delta \in (0,1]\), if the sample size satisfies
\begin{equation}\label{ineq:sample-complexity}
    m \geq C \left(\frac{H}{\alpha}\right)^2 \left(d_{\mathcal{H}} + \log\left(\frac{2}{\delta}\right)\right),
\end{equation}
where \(C\) is an absolute constant, then with probability at least \(1 - \delta\) over an i.i.d. sample \((\pi_1, \ldots, \pi_m) \sim \mathcal{D}^m\), the following uniform bound holds:
\[
    \sup_{h \in \mathcal{H}} \left|\Exp_{\pi \sim \mathcal{D}}\!\left[h(\pi)\right] - \frac{1}{m} \sum_{i=1}^m h(\pi_i)\right| \leq \alpha.
\]
\end{lemma}
\end{samepage}

\begin{proof}
Let \(\mathcal{G} = \{h/H : h \in \mathcal{H}\}\). Its fat-shattering dimension at every positive scale is at most \(d_{\mathcal{H}}\). By \citet[Theorem~1]{mendelson2003entropy}, its \(L_2(Q)\) covering number satisfies
\[
  \log N(t,\mathcal{G},L_2(Q))
  \le C_0 d_{\mathcal{H}}\log(2/t)
  \qquad (0<t<1)
\]
for every probability measure \(Q\).
For a sample \(S=(\pi_1,\dots,\pi_m)\), let \(Q_S\) be its empirical measure and let \(\sigma_1,\dots,\sigma_m\) be independent uniform signs in \(\{-1,1\}\). Dudley's entropy integral bounds the empirical Rademacher complexity as
\[
  \begin{aligned}
    \widehat{\mathfrak R}_S(\mathcal{G})
    &\triangleq \Exp_\sigma\!\left[\sup_{g\in\mathcal{G}}\frac1m\sum_{i=1}^m\sigma_i g(\pi_i)\right]\\
    &\le \frac{C_1}{\sqrt m}\int_0^1\sqrt{\log N(t,\mathcal{G},L_2(Q_S))}\,dt
    \le C_2\sqrt{\frac{d_{\mathcal{H}}}{m}},
  \end{aligned}
\]
since \(\int_0^1\sqrt{\log(2/t)}\,dt\) is finite.
Symmetrization, bounded differences, and a union bound over the two one-sided deviations give a uniform deviation of at most \(C_3 H\sqrt{(d_{\mathcal{H}}+\log(2/\delta))/m}\) with probability at least \(1-\delta\)~\citep{bartlett2002rademacher}.
All constants are absolute. The sample size in \cref{ineq:sample-complexity} makes this bound at most \(\alpha\).
\end{proof}

The following lemma records the ERM consequence used in the main text.

\begin{restatable}{lemma}{RestateCoroErmLearning}\label{coro:erm-learning}
    Fix \(H>0\), \(\alpha > 0\), \(\delta \in (0,1]\), instance space \(\Pi\), and cost function \(\mathtt{COST}:\mathcal{A}\times\Pi\to[0,H]\).
    Suppose the algorithm class \(\mathcal{A}\) induces a cost function class with pseudo-dimension at most \(d_{\mathcal{A}}\).
    If \(m\) satisfies \cref{ineq:sample-complexity} with \(d_{\mathcal H}=d_{\mathcal A}\), then with probability at least \(1-\delta\), any empirical risk minimizer \(\hat A\) satisfies
    \[
      \Exp_{\pi\sim\mathcal D}[\mathtt{COST}(\hat A,\pi)]
      -
      \min_{A\in\mathcal A}\Exp_{\pi\sim\mathcal D}[\mathtt{COST}(A,\pi)]
      \le 2\alpha.
    \]
\end{restatable}

\begin{proof}
Let \(\mathcal{H} = \{ h_A : \Pi \to [0,H] \mid A \in \mathcal{A} \}\), where \(h_A(\pi) := \mathtt{COST}(A,\pi)\).
By assumption, \(\mathrm{Pdim}(\mathcal{H}) \le d_{\mathcal{A}}\).
If \(m\) satisfies \cref{ineq:sample-complexity}, then \cref{lemma:uniform-convergence} gives, with probability at least \(1-\delta\),
\[
  \sup_{A \in \mathcal{A}} \Big| \Exp_{\pi\sim\mathcal{D}}[h_A(\pi)] - \tfrac{1}{m}\sum_{i=1}^m h_A(\pi_i) \Big| \le \alpha.
\]
Let \(A_{\mathcal{D}} \in \arg\min_{A\in\mathcal{A}} \Exp_{\pi\sim\mathcal{D}}[h_A(\pi)]\) be the optimal algorithm, and let \(\hat{A}\) be the ERM solution minimizing \(\tfrac{1}{m}\sum_{i=1}^m h_A(\pi_i)\).
Then
\[
    \begin{aligned}
        \Exp_{\pi\sim\mathcal{D}}[h_{\hat{A}}(\pi)]
        &\le \tfrac{1}{m}\sum_{i=1}^m h_{\hat{A}}(\pi_i) + \alpha\\
        &\le \tfrac{1}{m}\sum_{i=1}^m h_{A_{\mathcal{D}}}(\pi_i) + \alpha\\
        &\le \Exp_{\pi\sim\mathcal{D}}[h_{A_{\mathcal{D}}}(\pi)] + 2\alpha,
    \end{aligned}
\]
where the first and third inequalities follow from uniform convergence, and the second from the definition of ERM.
\end{proof}

\section{Smoothness of \MAXKCSP{}}\label{sec:maxcsp-smoothness}
This section verifies the smoothness condition for the standard polynomial formulation of \MAXKCSP{}.

Every problem in \MAXSNP{} L-reduces to \MAXKCSP{} for some constant \(k\)~\citep{Papadimitriou1991APX, Papadimitriou1994Computational}.
Specific graph problems such as \MAXCUT{} admit low-degree polynomial formulations, as in \cref{example: maxcut}.

We represent \MAXKCSP{} as a smooth integer program of degree \(k\) by arithmetization.
Consider an instance with \(n\) decision variables \(x_1, \ldots, x_n \in \{0,1\}\) and a constraint set \(\mathcal C\).
Each constraint \(c\in\mathcal C\) is defined on a \(k\)-variable scope and is specified by a Boolean function \(c : \{0,1\}^k \to \{0,1\}\).

We encode each constraint as a polynomial.
For a constraint \(c\) on variables \((x_{i_1}, \dots, x_{i_k})\), define \(t_c(\x)\) as the sum of indicator polynomials over its satisfying assignments.
For each assignment \(a \in \{0,1\}^k\) with \(c(a) = 1\), the indicator polynomial is
\[
  \delta_a(\x) = \prod_{r=1}^k x_{i_r}^{a_r} (1 - x_{i_r})^{1 - a_r}.
\]
Then, \(t_c(\x) = \sum_{a : c(a) = 1} \delta_a(\x)\).
The objective is the sum over all constraints:
\[
  p(\x) = \sum_{c\in\mathcal C} t_c(\x).
\]
This function is a polynomial of degree at most \(k\).

To bound the smoothness parameter, let \(\mu\) be the maximum number of constraints defined on any single \(k\)-variable scope.
\begin{itemize}
    \item For \MAXCUT{} on simple graphs, any pair of vertices has at most one edge, so \(\mu=1\).
    \item For \MAXKSAT{}, multiple clauses may involve the same set of variables, for example $x_1 \lor x_2 \lor x_3$ and $\neg x_1 \lor x_2 \lor x_3$. If duplicate identical clauses are excluded, then \(\mu \le 2^k\).
\end{itemize}

The following lemma bounds the smoothness parameter in terms of \(\mu\) and \(k\).

\begin{lemma}\label{lemma:maxkcsp-smoothness}
    The polynomial objective \(p(\x)\) derived from a \MAXKCSP{} instance is \(\beta\)-smooth for any \(\beta \ge \mu 2^k\).
\end{lemma}

\begin{proof}
    We bound each monomial coefficient in the expansion of \(p(\x)\).

    For any single constraint \(c\) on a set of \(k\) variables, the polynomial \(t_c(\x)\) is a sum of at most \(2^k\) terms of the form \(\delta_a(\x)\).
    The expansion of each \(\delta_a(\x)\) produces monomials with coefficients \(\pm 1\).
    Hence every monomial coefficient contributed by a single constraint has magnitude at most \(2^k\).

    A monomial on a set \(S\) of \(\ell\le k\) variables receives contributions only from constraints whose scope contains \(S\).
    There are \(\binom{n-\ell}{k-\ell}\le n^{k-\ell}\) such scopes, each carrying at most \(\mu\) constraints, so the coefficient of this monomial has magnitude at most \(\mu2^k n^{k-\ell}\le\beta n^{k-\ell}\).
\end{proof}

\end{document}